\documentclass[12pt]{article}
\usepackage[square,semicolon,authoryear]{natbib}
\usepackage{graphicx}
\usepackage{ntheorem}
\theoremseparator{.}
\usepackage{mathtools}
\usepackage{booktabs}
\newtheorem{thm}{Proposition}
\usepackage{caption}
\newcommand{\pD}{\partial}
\newcommand{\E}{{\mathbb E}}
\newcommand{\dd}{{\mathrm d}}

\usepackage{amsmath,amssymb,amsfonts,color}
\usepackage{mathrsfs}

\theoremstyle{nonumberplain}
\theorembodyfont{\normalfont}
\newtheorem{proof}{Proof}

\usepackage{geometry}
\title{The inverse Cox-Ingersoll-Ross process for parsimonious financial price modeling}
\date{}

\author{L. Lin{\small$^{\mbox{\,\ref{ECUST}~\footnote{Email: llin@ecust.edu.cn}}}$} 
~and~ 
D. Sornette{\small$^{\mbox{\,\ref{RisksX},\ref{SFI}
~\footnote{Email: dsornette@ethz.ch}
}}$} 
}%

\begin{document}
\maketitle
\vspace{-9mm}
\begin{enumerate}
\item{School of Business, East China University of Technology and Science, 200237 Shanghai, China}\vspace{-3mm}\label{ECUST}%
\item{Institute of Risk Analysis, Prediction and Management (Risks-X), Southern University of Science and Technology (SUSTech),
Shenzhen, 518055, China}\label{RisksX} \vspace{-3mm}%
\item{Swiss Finance Institute, c/o University of Geneva, 40 blvd. Du Pont d'Arve, CH 1211 Geneva 4, Switzerland}\label{SFI}\vspace{-3mm}%
\end{enumerate}

\abstract{
We construct new classes of financial price processes based on the insight that the key variable driving prices $P$ is the earning-over-price ratio $\gamma \simeq 1/P$, which we refer to as the earning yield and is analogous to the yield-to-maturity of an equivalent perpetual bond. This modeling strategy is illustrated with the choice for real-time $\gamma$ in the form of the Cox-Ingersoll-Ross (CIR) process, which allows us to derive analytically many stylised facts of financial prices and returns, such as the power law distribution of returns, transient super-exponential bubble behavior, and the fat-tailed distribution of prices before bubbles burst. Our model sheds new light on rationalizing the excess volatility and the equity premium puzzles. The model is calibrated to five well-known historical bubbles in the US and China stock markets via a quasi-maximum likelihood method with the L-BFGS-B optimization algorithm. Using $\phi$-divergence statistics adapted to models prescribed in terms of stochastic differential equations, we show the superiority of the CIR process for $\gamma_t$ against three alternative models.}

\vskip 0.5cm
\vspace{3mm} 


\section{Introduction}

We introduce a new approach to understand the properties of stock market prices, inspired by an apparently trivial decomposition
\begin{equation}\label{E:decomp1}
	P = \dfrac{E}{\gamma}~~,~~ \text{with} ~\gamma:=\dfrac{E}{P},
\end{equation}
where $P$ denotes the price of a given stock, $E$ is the corresponding current earnings per share (EPS)
and we refer to $\gamma$ as the earning yield, which is also called the earnings-over-price ratio. 
Note that $\gamma$ is the inverse of the widely utilised Price-to-Earnings (PE) ratio. 
Though less popular than the PE ratio, the earning yield is frequently encountered in financial statement analysis, because it is used as 
a metric for return on investment (ROI). As its name suggests, the earning yield represents how many monetary units (MU) a company earned for 1 MU invested in the company. The term \emph{earning yield}  may cause confusion for investors, as one might expect a company with a high earning yield to be better than one with a low earning yield. In reality, a company with the same earnings but a smaller earning yield is considered to have higher growth potential and is thus priced higher. In fact, the meaning of earning yield can be understood as (1) the relative importance of current earnings in pricing, and (2) the interest rate that a company issuing stock would face if it issued perpetual bonds instead, under the condition of surrendering growth and retention \footnote{ Section 2 will provide a more detailed discussion.}. Understandings these two facets can help disambiguate the meaning of earning yield and dispel any confusion. 

From a dynamic perspective, the expression (\ref{E:decomp1}) can be interpreted in  two different ways, leading to two distinct viewpoints on what dominates. The first interpretation focuses on the fact that the price $P$ varies linearly with earnings per share (EPS), given 1/$\gamma$, namely $P\simeq E$, while the second interpretation emphasises that the price $P$ varies inversely to the earning yield $\gamma$, namely $P\propto 1/\gamma$.  Given that the EPS is in general slowly varying with low volatility, we adopt the second interpretation and argue that the properties of the price process mainly arise from those of the earning yield.  
Our proposal is that it is more natural and illuminating to 
model the \emph{real-time} $\gamma$, interpreted as the instantaneous 
market collective beliefs on earning yield, rather than the \emph{real-time} price $P$ itself.  
Our strategy involves using time-dependent continuous stochastic processes to represent real-time $\gamma$ and we have discovered that, even with simple processes,  various empirically observed stock price and return properties can be rationalized, 
 such as the power law distribution of returns,
 transient bubble behavior, and the fat-tailed distribution of prices before bubbles burst. 
 Moreover, simple real-time $\gamma$ dynamics provide straightforward explanation of the excess volatility and the equity premium puzzles.

Indeed, the structure of equation (\ref{E:decomp1}) provides a natural route to account for the
excess volatility puzzle identified by Shiller and LeRoy for equity markets, which refers to the fact that
 stock prices frequently undergo large changes that are difficult to explain from the 
 generally small changes in prospective corporate earnings or dividends 
 \citep{shiller1981,LeRoy1981,shiller1989,lyons2001,handbookExchangeRates2012}.
Translated using the notations of equation (\ref{E:decomp1}), 
$E$ moves slowly with a much smaller variance than that of $P$.
This mechanically implies that the interesting or relevant price dynamics lies in the denominator $\gamma$.
The fact that $\gamma$ is small in general (of the order of a few percent) and is in the denominator provides a natural
way to understand structurally the excess volatility puzzle, as relatively small changes (in absolute value) of $\gamma$
translate into large changes in price $P$. As an often well-documented effect, suppose $E=10$ monetary units (MU)
and $P=250$ MU, corresponding to $\gamma=0.04$. Suppose that $\gamma$ decreases to $0.03$. Then, the price
increases to $333$, a 33\% change, while $E$ has remained unchanged.
The excess volatility puzzle thus provides a motivation for proposing price processes deriving from those of $\gamma$.
We start from this empirical observation and propose a class of models for the price process of the form (\ref{E:decomp1}),
 so that the large fluctuations of the price can be ascribed to a suitably defined dynamics of real-time $\gamma$.
 
 The main model we consider assumes that the real-time $\gamma$ is not fixed but fluctuates around some $\gamma^* >0$.
 We show that this leads to an increased long-term return for the investor, which is approximately 
 proportion to the variance of real-time $\gamma$. Again, this is the interplay between being in the denominator 
 of the price equation and being relatively small with non-negligible excursions or fluctuations 
 that provides a possible concise description of the
equity premium puzzle \citep{MehraPres85,MehraPres03,Kocherla96}, 
namely the observation that the compensation for holding risky stocks is abnormally large.
In particular, our approach makes clear that models assuming constant discount rates 
tend to underestimate future price growth compared with models that account for fluctuating mean-reverting discount rates.

Stock prices often exhibit sharp rises, followed by crashes and drawdowns, 
demonstrating strongly nonlinear or abrupt long-term mean-reversal characteristics. 
In other words, over time scales from years to a few decades, prices can transiently deviate a lot from their very long-term (logarithmic) trends.
However, prices eventually return to the very long term trend associated with the growth rate of the underlying economy over very long time scales
\citep{MSCIlongerm}.
This very long-term mean-reversion property is the result of the combination of successions of 
faster price regime growth corrected by large drawdowns \citep{SornetteCauwelsbub2015}.
This nonlinear mean-reverting nature of stock prices suggests that the dynamics of real-time $\gamma$ should follow a mean-reversal process.
Ensuring positivity is important for real-world applications, and the simplest mean-reversal process that does so is the CIR process \citep{CIR85}. 
Thus, as the first serious incarnation of our proposition to capture the complexity of the price dynamics
via the earning yield process, we use the CIR process to model the real-time $\gamma$.
This choice is natural when remembering that the earning yield has the structure of an interest rate, 
and the CIR process is one of the most popular interest rate models.

The remainder of this paper is structured as follows. Section 2 presents the mathematical structure of the model and discusses the no-arbitrage condition. In Section 3, we derive various properties of our proposed price dynamics, including detailed statistical results on the distributions of prices and returns, ergodicity, emergent risk premium, and the nature of transient super-exponential growth that may describe the phenomenon of bubbles, which are seen as ubiquitous in financial markets. Section 4 presents our empirical analysis of the model. First, we describe the algorithm used to calibrate the model, which is then applied to five well-known historical bubbles in the US and Chinese stock markets. We also make an empirical comparison of our model to different alternative earning yield processes using several distance diagnostic metrics. Finally, Section 6 concludes.

\section{Price and discount rate processes}

\subsection{From price to earning yield and back}

Defining the earning yield by
\begin{equation}\label{defgamma}
	\gamma := \dfrac{E}{P}~,
\end{equation} 
by definition, one can write 
\begin{equation}\label{pricing1}
	P=\dfrac{E}{\gamma}~.
\end{equation} 

This describes the driving forces acting on the stock price $P$ in terms of $E$, the  
current earnings per share (EPS) and the earning yield $\gamma$. The variable $\gamma$ can be interpreted in two ways. In the first interpretation, 
it represents 
the weight of current earnings in pricing the stock. In the second interpretation, 
$\gamma$ stands for the 
interest rate that a company should have in a scenario with no-growth and no-retention.
\begin{itemize}
	\item {\it First interpretation}: The earning yield represents the weight of a company's current earnings in the stock price. Investors take into account not only the current earnings but also the future (discounted) earnings, and the current stock price is derived from the combined contributions of both. In valuing a stock, they weigh the importance of current and future earnings. A lower earning yield signifies that investors place more emphasis on future earnings and/or adopt a smaller discount rate when pricing the stock. This understanding of the earning yield can thus rationalize why a lower earning yield indicates a higher growth potential and/or a lower discount rate.
      \item  {\it Second interpretation}: The earning yield of a stock is the yield to maturity (YTM) derived from viewing the stock as an equivalent perpetual bond and pricing it using the bond pricing methodology under the no-arbitrage condition. In other words, the earning yield is an \emph{interest-like} variable. 
Let us illustrate this point with the following simple example. Let us assume that the company's earnings are expected to grow at a fixed rate $g$, the retention ratio remains at $b$ each year, and the cost of capital for the stock is $r$. With the current EPS being $E$,  the Gordon-Shapiro stock pricing model \citep{Williams1938,gordon1956capital} or the  \cite{miller1961dividend} model both gives
\begin{align}\label{E:simple}
	P = \sum_{t=1}^{\infty}\dfrac{E(1-b)^t(1+g)^t}{(1+r)^t} = \dfrac{E}{\frac{1+r}{(1-b)(1+g)}-1} ~.
\end{align} 
With the definition of earning yield, we obtain 
\begin{equation}
\gamma=\dfrac{1+r}{(1-b)(1+g)}-1~. 
\end{equation}
Thus, a smaller earning yield implies a larger growth potential and/or smaller discount rate.
For a perpetual bond that pays the dividend $E$ per period, if it is issued at par value based on the stock, 
what should be its yield-to-maturity? According to eq.(\ref{E:simple}) and noting that the perpetual bond has par value $E/\gamma$, the
answer to the question is obtained by solving the following equation
\begin{equation}
	\dfrac{E}{\gamma}=\dfrac{E}{1+\text{YTM}}+\dfrac{E}{(1+\text{YTM})^2}+\cdots+\dfrac{E}{(1+\text{YTM})^n}+\cdots~,
\end{equation}
whose solution is YTM$=\gamma$.  
This analysis does not imply that a stock can be converted into a perpetual bond. Rather, it implies that, under no-arbitrage conditions, stock (equity) can be equated to bond (debt). In line with the well-known \cite{modigliani1958cost} theory, an adjustment of the capital structure by replacing equity with debt will not affect the value of the company (in the absence of taxes, bankruptcy costs, agency costs, and asymmetric information, and in an efficient market). 
In this simplified framework, the above analysis suggests that the earning yield of a company having growth potential and the cost of debt (yield to maturity, YTM) of a company that use earnings as interest to repay its debt are essentially two sides of the same coin. Hence, the earning yield can be further understood as the interest rate the company would face when growth and retention are surrendered.
\end{itemize}

In practice, the earning $E$ typically changes only on a quarterly or monthly basis, while prices can be recorded
to fluctuate at the millisecond time scale. Therefore, the real-time $\gamma$ (denoted by $\gamma_t$ hereafter) calculated using equation (\ref{defgamma}) on real-time prices  $P_t$, cannot be directly regarded as earning yield, but rather as the \emph{instantaneous collective market belief on earning yield}. Intuitively, for the first interpretation, $\gamma_t$ represents the collective belief on the weight of current earnings in pricing the stock.

Let us elaborate further on  $\gamma_t$ using the second interpretation.
Consider a company whose earnings follow the growth rate sequence $\{g^1_t,g^2_t, g^3_t, ...,g^n_t, ...\}$ 
at future times $t+1, t+2, ..., t+n, ...$ expected at time $t$.
Let us assume that the earnings are retained according to the retention rate sequence $\{b^1_t,b^2_t, b^3_t, ...,b^n_t, ...\}$, 
and that the required return of capital in each future period is the sequence $\{r^1_t,r^2_t, r^3_t, ...,r^n_t, ...\}$. In addition, the expected earnings per share in period $t$ is denoted $E^e_t$. According to 
the stock pricing model of \cite{gordon1956capital} or to the \cite{modigliani1958cost} model, the stock price is given by
\begin{align*}
	P_t&=\sum_{k=1}^{\infty}\dfrac{E_{t+i}^e-b^{i}_t\cdot E_{t+i}^e}{\prod_{i=1}^{k} (1+r_t^{i})}=\sum_{k=1}^{\infty}\dfrac{E_t\cdot \prod_{i=1}^k (1-b_t^i)(1+g_t^i)}{\prod_{i=1}^{k} (1+r_t^{i})}\\
	&=\sum_{k=1}^{\infty}\dfrac{E_t}{ \prod_{i=1}^k \frac{1+r_t^{i}}{(1-b_t^i)(1+g_t^i)}}~.
\end{align*} 

The last equation can be rewritten in the following form, 
\begin{align}\label{E: bbb}
P_t =   {E_t \over 1+\gamma^1_t}  + {E_t \over (1+\gamma^1_t)(1+\gamma^2_t)} + {E_t \over (1+\gamma^1_t)(1+\gamma^2_t)(1+\gamma^3_t)} + ...~,
\end{align}
with the recursion structure, 
\begin{align*}
	\gamma_t^1&= \dfrac{1+r_t^1}{(1-b_t^1)(1+g^1_t)}-1,\\
	\gamma_t^2&=  \dfrac{(1+r_t^1)(1+r_t^2)}{(1-b_t^1)(1-b_t^2)(1+g^1_t)(1+g^2_t)}\cdot\dfrac{1}{1+\gamma_t^1}-1,\\
	\vdots\\
	\gamma_t^n&=\prod_{i=1}^n \frac{1+r_t^{i}}{(1-b_t^i)(1+g_t^i)}\prod_{i=1}^{n-1}\dfrac{1}{1+\gamma_t^i}-1,\\
	\vdots
\end{align*}

Given a total number $n_s$ of stock shares, the total value of the company should be $V_t=P_t\, n_s$. 
Let us imagine that the company alters its capital structure by exchanging all its equity with debt, 
i.e., issuing as much bond as possible to finance its operations such that 
the earning in each period will just cover the required debt payout. This amounts to sacrificing growth and not retaining earnings in the future. 
According to the \cite{modigliani1958cost} theory, the total liability of the company under this new capital structure 
is equal to the total value of the company before the exchange of equity into debt, 
which is $V_t=P_t \,n_s$. Multiplying both sides of eq. (\ref{E: bbb}) by $n_s$, we have,
\begin{align}\label{E: bbb2}
V_t=   {n_s\, E_t \over 1+\gamma^1_t}  + {n_s\, E_t \over (1+\gamma^1_t)(1+\gamma^2_t)} + {n_s\, E_t \over (1+\gamma^1_t)(1+\gamma^2_t)(1+\gamma^3_t)} + ...~.
\end{align}
The condition of no-arbitrage requires that each term in the sum in the r.h.s. of eq.(\ref{E: bbb2}) represents the value of a zero coupon bond 
maturing at the corresponding horizon indicated by the largest superscript. If the liability comes entirely from issuing perpetual bonds, with yield to maturity 
$\gamma_t$, the value of the company becomes, 
\begin{align*}
	V'_t =\dfrac{n_s E_t}{1+\gamma_t}+\dfrac{ n_s E_t}{(1+\gamma_t)^2}+\cdots+\dfrac{n_s E_t}{(1+\gamma_t)^n}+\cdots~.
\end{align*}
With $V_t=V'_t$, this recovers 
\begin{equation}\label{E:pricing1t}
	\gamma_t=\dfrac{E_t}{P_t}~,
\end{equation}
which is the definition of real-time $\gamma$.  The correspondence between $V_t$ and $V'_t$ is similar
that that between forward rates and yield rates in bonds. Indeed, the set $\{ \gamma^1_t, \gamma^2_t, ..., \gamma^n_t,...\}$
plays the role of forward rates, which can be mapped onto a single yield rate $\gamma_t$.
From the equation (\ref{E:pricing1t}), one arrives at 
\begin{equation}\label{E:pricing1tt}
	P_t=\dfrac{E_t}{\gamma_t}~.
\end{equation}

We consider expression (\ref{E:pricing1tt}) as the price process.  To simplify the exposition and without significant impact on our results, since the volatility of $P_t$ is mostly driven by $\gamma_t$, we will take $E_t\equiv E$ as a constant. By doing so, the price process  becomes specified once the dynamics of $\gamma_t$ is prescribed. Hence, we focus on the evolution of $\gamma_t$, which we term the earning yield process.

\subsection{Two examples for earning yield processes}

\subsubsection{Geometric Brownian motion}

The simplest process that ensures that $\gamma_t$ remains positive is the geometric Brownian motion (GBM).
We can state the following proposition.
\begin{thm} \label{proposition0}
Given definition \eqref{pricing1}, let us assume that the earning yield is given by the geometric Brownian motion
\begin{equation}
	\dfrac{ \dd \gamma_t}{\gamma_t} =-(\mu-\sigma^2) ~ \dd t- \sigma ~\dd W_t~. \label{E:gbmgamma} 
\end{equation}
Then, the dynamics of the price is the solution of the following stochastic differential equation (SDE) representing 
also a geometric Brownian motion:
\begin{equation}\label{E:gbmprice}
	\dfrac{\dd P_t}{P_t}=\mu~ \dd t+\sigma ~\dd W_t~.
\end{equation}	
\end{thm}

 \begin{proof}
	We simplify the notations by removing the subscript $t$. According to It\^{o}'s lemma, expression (\ref{pricing1}) leads to
\begin{align*}
	\dd P &=\dfrac{\pD P}{\pD \gamma}\dd \gamma +\dfrac{1}{2}\dfrac{\pD^2 P}{\pD \gamma^2}(\dd \gamma)^2= -\dfrac{E}{\gamma^2}~\dd \gamma +\dfrac{1}{2}\cdot 2\cdot \dfrac{E}{\gamma^3}~(\dd \gamma)^2 \\
&=-\dfrac{E}{\gamma^2}\cdot [\,-(\mu-\sigma^2) \gamma~\dd t-\sigma\gamma~\dd W \,]+\dfrac{E}{\gamma^3}\sigma^2\gamma^2~\dd t\\
&=\mu\dfrac{E}{\gamma}~\dd t +\sigma \dfrac{E}{\gamma}~\dd W=\mu P ~\dd t +\sigma P~\dd W~.
\end{align*}
which recovers (\ref{E:gbmprice}).
Another way to derive this result is obtained by using the explicit solution of (\ref{E:gbmgamma}), which reads
\begin{equation}
\gamma_t = \gamma_0 e^{-(\mu-{1 \over 2} \sigma^2) t + \sigma W_t}~.
\end{equation}
Inserting in definition (\ref{pricing1}) gives
\begin{equation}
	P_t=\dfrac{E}{\gamma_t}=  P_0 e^{(\mu-{1 \over 2}\sigma^2) t + \sigma W_t}~,
	\label{thnrwynbr2gq}
\end{equation}
where $P_0 := \dfrac{E}{\gamma_0}$. 
Expression (\ref{thnrwynbr2gq}) is indeed the solution of equation (\ref{E:gbmprice}).
{\hfill{$\square$}}
\end{proof}

As a simple process that ensures positivity, the geometric Brownian motion is widely accepted 
as a convenient price dynamics benchmark \citep{samuelson2016proof}. 
However, despite its popularity coming from its simplicity, the geometric Brownian motion fails 
to account for many financial stylised facts  \citep{econoph-styl1,econoph-styl2}.
Proposition  \ref{proposition0} provides an interesting interpretation of why the 
geometric Brownian motion model for prices is financially unrealistic.
Indeed, for the price to grow (with a positive drift coefficient $\mu-{1 \over 2}\sigma^2$), the corresponding $\gamma_t$ 
has to have a negative drift coefficient. This implies that the weight of current earnings in 
the determination of pricing continuously decreases geometrically over time 
or that the interest rate of corporate financing continuously decreases geometrically over time.
This is strongly counterfactual. It is only for the case where $\mu={1 \over 2}\sigma^2$ that
both $P_t$ and $\gamma_t$ have zero drift.

\subsubsection{Cox-Ingersoll-Ross process}

Given the inadequacy of the geometric Brownian motion, we now study the case 
where $\gamma_t$ follows the Cox-Ingersoll-Ross (CIR) process,  which was initially proposed to model interest rates (see \cite{CIR85}):
\begin{equation}\label{CIRgamma}
	\dd \gamma_t=-\alpha~(\,\gamma_t-\gamma^*\,)~\dd t +\psi \sqrt{\gamma_t}~\dd W_t~.
\end{equation}
The CIR process has a mean-reversal property, so that $\gamma_t$ tends to fluctuate around a mean value $\gamma^*$
with the amplitudes of its deviations from $\gamma^*$ controlled by parameter $\psi$.
The mean-reversal property of $\gamma_t$ can be rationalized by a behavioral interpretation. The ``mean" can be explained as 
 the typical representative belief on earning yield of a stock in the market. The instantaneous collective market belief may not always equate to the market typical representative belief  due to two effects. First, short-term agreements may not fully reflect the beliefs of most traders. Second, not all traders are involved in trading at any given time, which can result in a noisy divergence and a transition from the short-run value of $\gamma_t$ to the long-run value $\gamma^*$.

The structure of the stochastic component proportional to $\sqrt{\gamma_t}$ ensures that $\gamma_t$ never crosses $0$,
remaining non-negative at all times. As we shall prove below, intuitively, this property of the CIR process 
ensures the stationarity of the price process. The value of $\gamma^*$ can be associated with a fundamental value
$P^* := E / \gamma^*$. 
The use of the CIR process to model the stochastic dynamics of $\gamma_t$ can be justified from the fact that 
$\gamma_t$ plays the role of interest rate.  


Using a process such as (\ref{CIRgamma}) inserted in (\ref{pricing1}) means that prices are discounted at time $t$
inter-temporally with the same discount rate $\gamma_t$. In other words, at each time $t$, investors form an anticipation of 
what is the correct interest rate (similar to the yield-to-maturity or YTM) to apply to future earnings at all future times. 
Here, $\gamma_t$ stands for the YTM of an equivalent perpetual bond, which is time-varying. The fact that $\gamma_t$
changes with $t$ expresses the changing anticipation of investors from one period to the next.

This representation should be distinguished from the model in which the term structure of discount rate is made 
explicit in terms of a set
$\{\widetilde{\gamma}^1_t,\widetilde{\gamma}^2_t,\cdots, \widetilde{\gamma}^n_t,\cdots\}$ so that the price is given by
\begin{equation}
P_t =   {E \over 1+\widetilde{\gamma}^1_t}  + {E \over (1+\widetilde{\gamma}^1_t)(1+\widetilde{\gamma}^2_t)} + {E \over (1+\widetilde{\gamma}^1_t)(1+\widetilde{\gamma}^2_t)(1+\widetilde{\gamma}^3_t)} + ...
\label{trhwbgvq}
\end{equation}
Equation (\ref{trhwbgvq}), which is similar to eq.(\ref{E: bbb}),  expresses the idea that the price $P_t$ is obtained by discounting future earnings with discount rates $\widetilde{\gamma}^n_t$ that are
anticipated to be varying along the term structure $1, 2, 3, ...$ and changing randomly from one period $t$ to the next.
Equation (\ref{trhwbgvq}) differs from (\ref{E: bbb}) in that the former corresponds to a loose pricing model where there may not necessarily be a connection between $\widetilde{\gamma}^k_t$ and $\widetilde{\gamma}^j_t$, whereas the latter assumes a discount model in which growth rate, retention ratio, and required rate of return are all anticipated and predetermined.
An end-member assumption is that the $\widetilde{\gamma}^n_t$'s are identically independently distributed (i.i.d.) according to some 
probability density function $p_{\widetilde{\gamma}}(\widetilde{\gamma})$. Then, performing the change of variable $x_n := {1 \over 1+\widetilde{\gamma}^n}$ and dropping 
the subscript $t$, equation (\ref{trhwbgvq}) becomes
\begin{equation}
P =   E \left( x_1 + x_1 x_2 + x_1 x_2 x_3+ x_1 x_2  x_3 x_4 +  ...\right)~,~~~~x_n := {1 \over 1+\widetilde{\gamma}^n}~,
\label{trhwbgwt2vq}
\end{equation}
where the $x_n$'s are i.i.d. random variables with probability density function $p_x(x)$.
In the sequel, we shall comment on the relationships between model (\ref{pricing1}) with (\ref{CIRgamma}) 
and model (\ref{trhwbgwt2vq}).

\begin{thm} \label{proposition1}
Given definition (\ref{pricing1}), let us assume that the earning yield is given by
\begin{equation}
	\dd \gamma_t=\alpha~(\gamma^* - \gamma_t)~\dd t +\psi\sqrt{\gamma_t}~\dd W_t~,\qquad \gamma^*=\dfrac{E}{P^*}~.\label{E:bubbleCIRgamma} 
\end{equation}
Then, the dynamics of the price is the solution of the following stochastic differential equation (SDE) 
\begin{equation}\label{E:Pb}
	\dfrac{\dd P_t}{P_t}=\left[~\alpha\left(\dfrac{P^*-P_t}{P^*}\right)+\psi^2\dfrac{P_t}{E}~\right]\dd t+\psi \sqrt{\dfrac{P_t}{E}}~\dd W_t~.
\end{equation}	
\end{thm}
The proof is in Appendix \ref{proof-E:Pb}.

\subsection{No-arbitrage condition}

 We show how the proposed price dynamics (\ref{E:Pb}) excludes risk-free arbitrage opportunities.
 In other words, we formulate a no-arbitrage condition.
 
Let us assume that the market is made of $N$ stocks. For each of these stocks, there is at least one derivative 
(e.g., a European call option) for which it is the underlying asset. We also assume 
the existence of a unique risk-free asset with a riskless interest rate $r_f$. The price dynamics of the $i$th stock price has the same
form as equation (\ref{E:Pb}),
\begin{equation}\label{E:Pbi}
	\dfrac{\dd P^{(i)}_t}{P^{(i)}_t}=\left[~\alpha^{(i)}\left(\dfrac{P^{(i)*}-P^{(i)}_t}{P^{(i)*}}\right)+\psi^{(i) 2}~\dfrac{P^{(i)}_t}{E^{(i)}}~\right]\dd t+\psi^{(i)} \sqrt{\dfrac{P^{(i)}_t}{E^{(i)}}}~\dd W^{(i)}_t~,
\end{equation}	
with $i=1,2,\cdots, N$. 

According to the fundamental theory of asset pricing, the opportunity of risk-free arbitrage for any asset trading on the market is absent if and only if there is at least one stochastic process $m_t$, the so-called stochastic discount factor (SDF), that satisfies the following two conditions (see \cite{harrison1979martingales, harrison1981martingales,duffie2010dynamic}): 
 \begin{itemize}
\item[(1)] $m_t$ is a martingale process, which means that $\E_t(\dd (m_t)=0$.
\item[(2)] The product of $m_t$ and the price discounted by the risk-free rate $r_f$ of any asset is also a martingale process: 
$m_t\,\widetilde{P}_t = \E_t(m_{t'}\,\widetilde{P}_{t'}), t'>t$, where $P_t$ is the observed price and
 $\widetilde{P_t}=e^{-r_f t} P_t$. This condition is equivalent to ${\E_t(\dd m_t P_t) \over m_t P_t}=r_f dt$.
\end{itemize}
By Ito calculus, 
\begin{equation}
{\E_t(\dd m_t P_t) \over m_t P_t} = {\E_t(\dd m_t) \over m_t}+{\E_t(\dd P_t) \over P_t}  +{\E_t(\dd m_t \dd P_t) \over m_t P_t}~.
\label{thwrtbgv}
\end{equation}
Since $m_t$ is a martingale, the second condition reduces to 
\begin{equation}\label{E: martingale2}
	\E_t\left(\dfrac{\dd P_t}{P_t}\right)+\E_t\left(\dfrac{\dd m_t}{m_t}\dfrac{\dd P_t}{P_t}\right) = r_f\dd t~.
\end{equation}
The following proposition states that there is at least one SDF $m_t$ that eliminates any opportunity of risk-free arbitrage for a market comprising $N$ stocks and their derivatives. 

\begin{thm} \label{proposition2}
Suppose a market contains $N$ stocks and their derivatives. The price $P^{(i)}_t$ of the $i$th stock is the solution of Equation (\ref{E:Pbi}) and the price of its derivative $\xi^{(i)}_t$ is determined by a $C^2$ bivariate function of $P^{(i)}_t$ and $t$, $f^{(i)} : (P^{(i)}_t , t) \mapsto \xi^{(i)}_t$. 
Consider the stochastic process $m_t$ given by	
\begin{equation*}
	\dfrac{\dd m_t}{m_t}=\vartheta^{(1)}_t\dd Z^{(1)}_t+\vartheta^{(2)}_t\dd Z^{(2)}_t+\cdots+\vartheta^{(N)}_t\dd Z^{(N)}_t, \quad i=1,\cdots,N,
\end{equation*}
where the following conditions hold:
\begin{align*}
    &\E_t[\,\dd W^{(i)}_t \dd Z^{(j)}_t\,] =\rho^{(i)}\,\delta_{ij}\dd t,\quad i, j=1,\cdots,N,\\
	&\vartheta^{(i)} _t =\dfrac{1}{\rho^{(i)}}\left(\dfrac{[r_f-\alpha^{(i)}]-\psi^{(i)2}}{\psi^{(i)}}\dfrac{1}{\sqrt{P_t^{(i)}E^{(i)}}}+\dfrac{\alpha^{(i)}\sqrt{P_t^{(i)}E^{(i)}}}{P^{*(i)}\psi^{(i)}}\right)\, \quad i =1,\cdots,N,
\end{align*}
where $\rho^{(i)}$ is an arbitrary multivariate function of $N$ independent variables,  $P^{(i)}_t, i=1,2,\cdot, N$, and $Z^{(i)}_t$, $W^{(i)}_t$ are Gauss-Wiener processes.
Then, $m_t$ qualifies as a SDF for the whole market, implying that the no-arbitrage condition holds. 
\end{thm}
The proof is in Appendix \ref{proof-martin-SDF}. 

It is important to stress that Proposition 2 only demonstrates the existence of a SDF in the market without referring to its uniqueness. 
Different specifications of $\rho^{(i)}$ will result in different qualified SDFs. 

The proof  in Appendix \ref{proof-martin-SDF} gives an intuition of the mechanism for the absence
of risk-free arbitrage opportunities. Consider a mini-market that comprises only the $i$th stock, some derivatives of this stock, and a risk-free bond. The proof takes advantage of the fact that any asset in the mini-market can be replicated from the other assets. This means that
each such mini-market can be devoid of risk-free arbitrage opportunities. Given that the entire market comprises $N$ such mini-markets, 
it stands to reason that there is no risk-free arbitrage opportunity for the entire market. 
Meanwhile, it is unsurprising that the universal SDF $m_t$ acts as the combination of the $N$ local SDF's of each mini-market.  

\subsection{Representing the earning yield process using the Campbell-Shiller Decomposition}

The common practice in finance is to work with the logarithm of prices or the logarithm of valuation ratios.
Using the logarithm of prices as a starting point makes no substantive
difference between the price and its inverse.  The question arises as to whether 
focusing on the inverse of the price instead of the log of the price, or
the log price-dividend ratio, is a substantive step forward. This section addresses this question
and argues that working with the inverse of the price makes a substantive difference.

Let us first start with the 
Campbell-Shiller decomposition, which is a well-known logarithmic linearization modeling approach that can represent log-price or its variant valuation ratio as a present value model of rational expectations for future fundamental changes and real discount rates \citep{csdecomp,campbelldecomp}. With the Campbell-Shiller decomposition, one can show that the logarithm of earning yield ($\ln\gamma_t$)  can also  approximately be characterized as a constant plus the ``present value" of expected current and  future values of the one-period required return minus the  one-period growth in earnings adjusted by the dividend ratio. 

By denoting the dividend ratio at time $t$ by $c_t$, which is $1$ minus the retention ratio $b_t$,  and denoting the dividend by $D_t$, we have,    
\begin{align}
	R_{t+1} =\dfrac{P_{t+1}+D_{t+1}}{P_t}=\dfrac{P_{t+1}+ E_{t+1}c_{t+1} }{P_t}
\end{align}
Then, the current price-earning ratio can be expressed by the following multiplicative form ex-post:
\begin{equation}
	\dfrac{P_t}{E_t}=\dfrac{1}{R_{t+1}}\left(1+\dfrac{P_{t+1}}{E_{t+1} c_{t+1}} \right)\cdot \dfrac{E_{t+1}}{E_t}\cdot c_{t+1}~.
	\label{trhwbgq}
\end{equation}
Taking the logarithm of both terms in equation (\ref{trhwbgq}) gives
\begin{equation}\label{E:lngamma}
	\ln \gamma_t=\ln R_{t+1}-\ln \left(1+\dfrac{1}{c_{t+1}\gamma_{t+1}}\right)-(\ln E_{t+1}-\ln E_t)-\ln c_{t+1}
\end{equation}

Starting from eq.(\ref{E:lngamma}), the Campbell-Shiller decomposition procedure involves employing a first-order Taylor approximation and iterative forward steps, assuming the validity of the transversal condition $\lim_{n\to\infty} \rho^n \ln\gamma_{t+n}=0$.  Using symbol $\mathscr{L}$ to denote the lead operator, and denoting $r_{t} :=\ln R_t$ and $e_t :=\ln E_{t}-\ln E_{t-1} +(1-\rho) \ln c_t$,
we obtain (See Appendix \ref{A:laow0nn} for the detailed derivation):
\begin{align}\label{E:csdecompforlng}
	\ln \gamma_t &=-\dfrac{\kappa}{1-\rho}+\dfrac{r_{t+1}}{1-\rho\mathscr{L}}-\dfrac{e_{t+1}}{1-\rho\mathscr{L}}\notag\\
	&=-\dfrac{\kappa}{1-\rho} + (1+\rho\mathscr{L}+\rho^2\mathscr{L}^2+\cdots)(r_{t+1}-e_{t+1})\\\notag
	&=-\dfrac{\kappa}{1-\rho} + \sum_{j=0}^{\infty} \rho^j (r_{t+1+j}-e_{t+1+j}), 
\end{align}
with 
\begin{equation}
	\begin{cases}
		\kappa = \ln(1+e^{-\overline{\ln c\gamma}}) + \dfrac{e^{-\overline{\ln c\gamma}}}{1+e^{-\overline{\ln c\gamma}}}	\cdot \overline{\ln c\gamma}\\
		\rho = \dfrac{e^{-\overline{\ln c\gamma}}}{1+e^{-\overline{\ln c\gamma}}}	~,
	\end{cases}
\end{equation}
where $\overline{\ln c\gamma}$ represents $\mathbb{E}(\ln c_t\gamma_t)$, the unconditional expected logarithm of the dividend-price ratio. 
Note that decomposition (\ref{E:csdecompforlng}) mathematically holds for any returns, prices, earnings and dividend ratio ex-post. 
But ex-ante, rational expectation is necessary to guarantee that the market precisely keeps all the future value $r_{t+1+j}$ and $e_{t+1_j},\quad j=0,1,2,...$ in mind when determining the current $\ln \gamma_t$. 

The decomposition (\ref{E:csdecompforlng}) serves as a potential starting point for modeling $\ln \gamma_t$. 
However, in the present paper, we choose not to adopt this approach for explaining and characterizing the behavior of $\gamma_t$. 
for the two following reasons.

\begin{enumerate}
\item  One of the conditions for the validity of the Campbell-Shiller  decomposition is that $\gamma_t$ should satisfy the transversal condition $\lim_{n\to\infty} \rho^n \ln\gamma_{t+n}=0 $. However, this condition effectively rules out the possibility of the existence of bubbles. As we will see soon, when the market is in a recurrent explosive bubble regime, $\gamma_t$ has a positive probability of rapidly dropping to zero at a certain finite time $t_c$ and instantaneously rebounding (correspondingly, the price surges to infinity and is instantly reflected). In this case, $\exists t_c < \infty,\text{where} \ln\gamma_{t_c}=\infty$ and thus the transversal condition does not hold. Moreover, when $\ln\gamma_{t_c}$ approaches infinity,  the current value of $\ln \gamma_t$ also diverges when forecast is in the context of rational expectations. Hence, the Campbell-Shiller decomposition method based on rational expectations completely fails in such bubble scenarios. In contrast, the modeling approach employed in this paper is more suitable for characterizing bubbles.
\item The Campbell-Shiller decomposition relies on a first-order Taylor expansion around $\mathbb{E}(\ln c_t\gamma_t)$ to specify the dynamics of $\ln \gamma_t$ , with the implicit assumption that the deviations of $\ln \gamma_t$  from $\mathbb{E}(\gamma_t)$ and $\ln c_t$ from $\mathbb{E}(\ln c_t$ are small. However, even when $c_t$ remains constant, a slight deviation $\Delta \gamma_t$ of $\gamma_t$ from its mean $\gamma^*$ will result in a deviation of $\ln \gamma_t$ roughly by $\dfrac{\Delta \gamma_t}{\gamma^*}$. Since $\gamma^*$ typically has magnitudes in the range $10^{-1} \sim 10^{-3}$, $\ln \gamma_t$ deviates significantly from  $\mathbb{E}(\gamma_t)$, thus undermining the accuracy of the first-order approximation.
\end{enumerate}

We now return to the general question of why focusing on the inverse of the price instead of the logarithm of the price is a substantive step forward. 
Our contribution is to propose that one should construct processes for the price in the form of the logarithm of a process 
motivated to model the inverse of the price and not start
from expressing the logarithm of the price as a process itself.
In this sense, our framework in term of the logarithm of the earning yield provides an important insight because it frames the search for price process
in a certain mathematical form. If the starting point is to imagine processes for the logarithm of the price, the 
corresponding natural class of models is very different.
It would not be trivial to imagine a priori that the logarithm of the  Cox-Ingersoll-Ross model is relevant.
In that sense, our framework is quite fundamental as identifying the ``right'' variable that makes the theory simple and insightful.

Our claim that the logarithmic of earning yield is the ``right variable'' actually emphasises, from an epistemological perspective, that  $P_t  \leftarrow \dfrac{E}{\gamma_t}$ instead of $\gamma_t \leftarrow \dfrac{E}{P_t}$.
The symbol ``$\leftarrow$" represents the direction of causality, indicating that the symbol on the right is the ``cause," while the symbol on the left is the ``effect." Although the cause and effect are equivalent in terms of quantity, they exhibit asymmetry. Traditional notation often carelessly uses ``=" to represent causality, which can easily lead to ambiguity. The equal sign only expresses equivalence in quantity, allowing for symmetrical treatment of variables on both sides, enabling them to be freely rearranged or subject to inverse functions.
However, in reality, cause and effect are asymmetrical and cannot be arbitrarily interchanged. Using ``=" to denote causality is a poor mathematical practice that misrepresents the relationship between variables. For example, Newton's second law is often written as $F=ma$, but it is actually expressing $a \leftarrow \dfrac{F}{m}$, meaning that the acceleration is determined by the external force, rather than $F \leftarrow ma$, which would imply that the amount of acceleration of a object determines the forces acting on the object. Similarly, directly writing $P_t ={E}/{\gamma_t}$ for  $P_t  \leftarrow {E}/{\gamma_t}$ can easily lead one to innocuously transform the equation into $\ln \gamma_t = \ln E - \ln P_t$. Consequently, it may be mistakenly assumed that modeling $-\ln P_t$ directly is equivalent to modeling $\ln \gamma_t$. Our contribution lies in recognizing that, in reality, $\gamma_t$ is the cause and $P_t$ is the effect, i.e., $P_t$ is determined by $\gamma_t$, not the inverse. Moreover, the underlying meaning of $\gamma_t$ is the inverse of price. Therefore, the inverse of price is a more essential determining variable than the price itself. As the example of Newton's second law illustrates, modeling the external force, the ``cause," provides a more concise and clear characterization and prediction of the acceleration than directly modeling the acceleration itself. Similarly, modeling $\ln \gamma_t$ instead of $\ln P_t$ achieves the same effect, which is the basis for the present work.

\section{Properties of the price process (\ref{E:Pb})} 

\subsection{Classification of two different regimes \label{bhrgrw}}

The properties of the price process (\ref{E:Pb}) derive from those of the CIR process (\ref{E:bubbleCIRgamma}).
We consider solely the regimes where $\alpha, \gamma$ and $\psi$ are strictly positive.
Two regimes need to be considered. 
\begin{enumerate}
\item {\it Non-explosive nonlinear regime}: for $2\alpha \gamma^* > \psi^2$, the CIR process $\gamma_t$ starting from an arbitrary initial positive value 
remains strictly positive at all times. The condition $\alpha>0$ then ensures that the CIR process is stationary and ergodic. 
Correspondingly, this implies that the price process remains bounded at all times and is also stationary and ergodic. This is proved
in Appendix \ref{A:ergodic}. 

\item {\it Recurrent explosive bubble regime}: for $2\alpha \gamma^* \leq \psi^2$, the CIR process $\gamma_t$ has a strictly positive probability of hitting $0$ and, when this occurs, the origin is 
instantaneously reflecting. Let $t_c=\inf \{t: \gamma_t=0\}$ be the time of first hit with the origin. Then, the price $P_t$
has a strictly positive probability of diverging in finite time at some stochastic time $t_c$. 
We call this regime the recurrent exploding bubble regime, characterised by stochastic finite-time singularities. 
The price process is also stationarity and ergodic at long times.
\end{enumerate}

A number of models of financial bubbles with stochastic finite-time singularities have been previously proposed 
\citep{SorJohBou96,JSL99,JLS00,AS1,IdeSor02,AS2fearful,CorsiSor14,LinRenSor14,lin2013,LinSchSor19,SchSor20}. The model closest to the present work
in the exploding bubble regime is \citep{AS1,AS2fearful} in which the price is an inverse power of a 
drifting Brownian motion. The exploding bubble regime generalises these works by considering the more elaborate
richer and more financially relevant CIR process. 

For fixed $\alpha$ and $\psi$, the condition $2\alpha \gamma^* \leq \psi^2$ implies that a too low reference interest rate $\gamma^*$ 
or a too weak mean reversal coefficient $\alpha$
spells potential disaster in the form a looming finite time singularity. Of course, singularities (almost) never
occur in Nature or Society but the dynamics preceding them can develop until a time when the associated excesses
trigger other mechanisms to take over and make the system transition to other phases, such as crashes, drawdowns
and even recessions. From this perspective, the policies of ultra-low interest rates taken by central banks in response to the great financial crises
can be interpreted as having pushed the reference interest rate $\gamma^*$ too low (below the safety zone 
threshold $\psi^2/2\alpha$), thus
promoting dangerously extraordinary price appreciations, found in real estate as well as in certain sectors
of the economy such as the digital social sector of GAFAM.

Another way to interpret this regime comes from the definition $\gamma^*=E/P^*$: if the anchoring price $P^*$
becomes too large compared with the earning $E$ (i.e., $P^*>H :=2\alpha E/\psi^2$), the price dynamics enters a fundamentally
unstable regime with a bubble developing strongly and exploding as a finite-time singularity.

The bifurcation from a stable price regime when $2\alpha \gamma^* > \psi^2$
to the explosive regime when $2\alpha \gamma^* \leq \psi^2$, just from a change of expectation of the reference price 
and/or earnings, provides a simple formal rationalisation of the common lore that financial 
markets are characterised by shifts between stable and unsustainable exuberant regimes.
For instance, \cite{Biancietaljf22} shows that the US economy can be characterised by longer-term regime shifts in asset values 
that synchronise with shifts in interest rates. 

While such strongly explosive price behaviours can sometimes 
be observed at least transiently during some financial bubbles, most bubbles seem to exhibit weaker explosive regimes \citep{Phillips2011}
but still characterised by transient ``super-exponential'' growth \citep{SornetteCauwelsbub2015,Ardila-AlvarezSor21}.
In the remainder of this work, we thus focus on the first regime $2\alpha \gamma^* > \psi^2$, which 
exhibits many interesting properties.
This condition $2\alpha \gamma^* > \psi^2$ means that the volatility should not be too large, or alternatively
the strength $\alpha$ of the mean reverting process as well as the amplitude of the typical interest rate $\gamma^*$ should be 
sufficiently large. 

\subsection{Transition probability of the price}

The price process characterised by Equations (\ref{pricing1}) and (\ref{E:bubbleCIRgamma}) has the following explicit transition density. 

\begin{thm}\label{proptranden}
Defining $P^* := \dfrac{E}{\gamma^*}$ and $H:=\dfrac{2\alpha E}{\psi^2}$,
the probability $f(P_t,t\mid P_0)$ to find the price value $P_t$ at time $t$ given that the price was $P_0$ at time $0$ is given by
\begin{equation}\label{E:traden}
	f(P_t,t\mid P_0)=c~e^{-(u+v)}v^{\frac{q}{2}+2}u^{-\frac{q}{2}}~I_q(2\sqrt{uv})~, 
\end{equation}
where
\begin{align}
c &=\dfrac{1-e^{-\alpha t}}{H}  \label{trhtgbq1} \\
u & = \dfrac{H}{P_0}\cdot\dfrac{1}{e^{\alpha t}-1} \label{trhtgbq2} \\
v & = \dfrac{H}{P_t}\cdot\dfrac{1}{1-e^{-\alpha t}} \label{trhtgbq3} \\
q & = \frac{H}{P^*}-1~,  \label{trhtgbq4}
\end{align}
and $I_q(\cdot)$ is the modified Bessel function of the first kind of order $q$, 
\begin{equation*}
	I_q(x)=\sum_{k=0}^{\infty}\left(\dfrac{x}{2}\right)^{2k+q}\dfrac{1}{k!~\Gamma(k+q+1)},\quad x\in\mathbb{R}.
\end{equation*}
\end{thm}
The proof of Proposition \ref{proptranden} is given in Appendix \ref{proof-transition-prob}. 

Let us note that the exponent $q$ remains positive for $P^* \leq H$, which is equivalent to $2\alpha \gamma^* \geq \psi^2$.
This is nothing but the condition that guarantees that the CIR process never touches $0$,
as discussed in section \ref{bhrgrw}. In contrast, for $P^* > H$, or $2\alpha \gamma^* < \psi^2$,
$q$ becomes negative, and the exponent $\frac{q}{2}+2$ becomes smaller than $2$, implying 
the mathematical divergence of the mean of the prices, as we shall discuss below.

\subsection{Stationary distribution of the price \label{rhjun4euj3n}}
 
 The fact that the price process defined by Equation \eqref{E:Pb} is an ergodic Markovian process 
derives from the equivalent property for the CIR process $\gamma_t$.
The proof of ergodicity is presented in Appendix \ref{A:ergodic}.  
There thus exists a stationary distribution $\pi(P)$, which is the long-time limiting distribution of the transition distribution (\ref{E:traden}).

\begin{thm}\label{prop:stableden}
Let us consider a price process characterised by the stochastic differential equation (SDE) \eqref{E:Pb}. 
Irrespective of whether $H \geq P^*$ or $H< P^*$  and for an arbitrary $P_0$, the invariant probability density of the price is given by      
\begin{equation}\label{invariantdentsiy}
	\pi(P)={H^{\mu^*} \over \Gamma \left(\mu^*\right)}e^{-\frac{H}{P}}P^{-(1+\mu^*)}~,~~{\rm with}~\mu^*=H/P^*= 2\alpha\gamma^*/\psi^2~.
\end{equation} 
The characteristic function of the price distribution is correspondingly
\begin{equation}
	\varphi_P(t)=\E(e^{it P})=\dfrac{2 H^{\frac{H}{2P^*}}}{\Gamma\left(\frac{H}{P^*}\right)}(-i\,t)^{\frac{H}{2P^*}}K_{\frac{H}{P^*}}(2\sqrt{H}\sqrt{-i\,t})~,
\end{equation}
where $K_m(\cdot)$ denotes the modified Bessel function of the first kind order $m$ (see \cite{cfIG} for the general derivation).
\end{thm}
The derivation of expression (\ref{invariantdentsiy}) is given in Appendix \ref{proof-stablede}.

Expression (\ref{invariantdentsiy}) shows that the probability density of the price has a power tail $\pi(P) \simeq 1/P^{1+\mu^*}$,
with a tail exponent $\mu^*:=H/P^*= 2\alpha\gamma^*/\psi^2$.  
The non-explosive nonlinear regime ($2\alpha \gamma^* > \psi^2; H>P^*$) corresponds to $\mu^*>1$ so that the mean of the price is exists.
In contrast, in the recurrent explosive bubble regime ($2\alpha \gamma^* \leq \psi^2; H<P^*$), $\mu^*\leq1$ and thus the mean of
the price is infinite. Thus, the transition from the non-explosive nonlinear regime to the recurrent explosive bubble regime
is mirrored by the transition from a price process with finite to infinite mean (see \citep{Sornette_CPNS04}
for a classification of power laws according to the values of the tail exponent $\mu^*$ and its consequences).

For arbitrary very long sequences of $N$ prices $\{\,P_{t_0},P_{t_1},\cdots,P_{t_N}\,\}$ at time-points $\{\,t_1,t_2\dots,t_N\,\}$,
these prices are asymptotically and identically distributed random variables with distribution density (\ref{invariantdentsiy}).
Thanks to the property of ergodicity, this stationary distribution of prices along time can be converted to an ensemble distribution 
or cross-sectional distribution (at a fixed time)  across all possible parallel price realisations.
Applied to a stock market of $N$ stocks considered statistically equivalent,  
the asymptotic cross-sectional tail dependence of stock prices can be described by a power-law distribution with tail exponent $\mu^*$. 
\cite{kaizoji2006} pioneered the diagnostic of the existence of global market bubbles via the transition of the power law tail exponent $\mu^*$
for the cross-sectional distribution of stock prices
from values larger than $1$ to a value converging to $1$ close to the end of the Japanese Internet bubble.
Within our model framework, the maturation of the Internet bubble can be interpreted as decreases of $\alpha, \gamma^*$ and/or increase of $\psi$
such that $P^*$ increased to overpass a decreasing $H$.
Given the empirical evidence that, in most financial bubbles, the crash follows a period of lower volatility
\citep{SornetteCauwels2018}, we infer that $\psi$ does not increase significantly during a bubble regime.
We thus deduce that it is $\alpha$ and/or $\gamma^*$ that decreased during the maturation of the Internet bubble.
In summary, the price process characterised by equations (\ref{pricing1}) and (\ref{E:bubbleCIRgamma}),
or equivalently by equation \eqref{E:Pb}, rationalises the curious observation of \cite{kaizoji2006} that had
not until now found a theoretical explanation.

\subsection{Conditional moments of price and returns}\label{S:condevr}

Based on the explicit conditional transition density (\ref{E:traden}), the first and second conditional moments of the price and cumulative return can be easily calculated. 

\begin{thm}\label{rhj243bq}
Defining the cumulative return $R_t=\dfrac{P_t}{P_0}$,
for the price process characterised by equations (\ref{pricing1}) and (\ref{E:bubbleCIRgamma}) and
for $H\geq P^*$, the conditional expectation of $R_t$ over the time interval from $0$ to $t$ is given by
\begin{equation}\label{E:condret}
	\E(R_t\mid P_0)=\dfrac{e^{-u}}{c\,q\,P_0}\cdot{}_1F_1(q,q+1,u)~,
\end{equation}
where the parameters $H$,$c$,$u$,$q$ are the same as in Proposition \ref{proptranden}, and ${}_1F_1(r,s,x)$ is a Kummer confluent hypergeometric function defined as
\begin{align*}
	{}_1F_1(r,s,x)=\dfrac{\Gamma(s)}{\Gamma(s-r)\Gamma(r)}\int_0^1 e^{xz}z^{r-1}(1-z)^{s-r-1}\dd z=\dfrac{\Gamma(s)}{\Gamma(r)}\sum_{n=0}^{\infty}\dfrac{\Gamma(n+r)}{\Gamma(n+1)\Gamma(n+s)}\cdot u^n~.
\end{align*}
\end{thm} 
The proof of this result is given in Appendix \ref{jiu5mori9}.  

\begin{thm}\label{wrhbtbqbq}
For the bubble process characterised by Equations (\ref{pricing1}) and (\ref{E:bubbleCIRgamma}), for $H\geq P^*$, 
the conditional variance of $R_t$ is given by
\begin{equation}\label{condvar}
\mathbb{V}ar(R_t\mid P_0)=\dfrac{e^{-u}}{q c^2P_0^2}\left[\dfrac{1}{q-1}\cdot {}_1F_1(q-1,q+1,u)-\dfrac{1}{q e^{u}}\cdot{}_1F_1(q,q+1,u)^2\right].
\end{equation}	
\end{thm}
The derivation of this result is given in Appendix \ref{A:condvar}. In a nutshell, 
we have
\begin{equation}\label{condp2}
\E(P_t^2\mid P_0)=\dfrac{e^{-u}}{c^2 q(q-1)}\cdot{}_1F_1(q-1,q+1,u).
\end{equation}
Using $\mathbb{V}\text{ar}(R_t\mid P_0):=\E(R^2_t\mid P_0)-\E(R_t\mid P_0)^2=\dfrac{1}{P_0^2}~[~\E(P^2_t\mid P_0)-\E(P_t\mid P_0)^2~]$,
this allows one to obtain (\ref{condvar}).

\begin{thm}\label{wrhwth2tybtbqbq}
Using expression (\ref{invariantdentsiy}) for the stationary distribution $\pi(P)$, we obtain
\begin{equation}
\E (R_{\infty}\mid P_0)=\dfrac{H\cdot\Gamma\left(\frac{H}{P^*}-1\right)}{P_0\cdot\Gamma\left(\frac{H}{P^*}\right)}=\dfrac{P^*}{P_0}\cdot\left(1-\dfrac{P^*}{H}\right)^{-1}~.\label{E:condlongR}
\end{equation}
and
\begin{align}\label{condsteadyvar}
\mathbb{V}ar(R_{\infty}\mid P_0) &=	\E (R_{\infty}\mid P_0)-\E (R_{\infty}\mid P_0)^2\notag\\ &=\dfrac{P^{*2}}{P^2_0}\cdot\left(1-\dfrac{P^*}{H}\right)^{-1}\left[\left(1-\dfrac{2P^*}{H}\right)^{-1}-\left(1-\dfrac{P^*}{H}\right)^{-1}\right]~.	
\end{align}
 \end{thm}
 The proof of these results are given in Appendix \ref{jiu52th1mori9}.  

These two expressions (\ref{E:condlongR}) and (\ref{condsteadyvar}) are valid only for $H>P^*$. 
For $H\leq P^*$, $\int_0^{\infty} e^{-z}~z^{\frac{H}{P^*}-2}~\dd z \to +\infty$ and $\E (R_{\infty}\mid P_0)=\infty$. 
This is just another way to retrieve the divergence of the mean of a random variable distributed 
according to a power law tail with exponent $\mu^* <1$ \citep{Sornette_CPNS04} as discussed above.
The long-term expectation of cumulative returns does not exist even if the long-term stationary density exists. 
The mechanism is simply that, as time passes, larger and larger realisations of $R_t$ are sampled so that the 
mean does not converge and diverges stochastically. Similarly, when $H\leq 2P^*$, 
the variance of the long-term return does not exist, as the integral $\int_0^{\infty} e^{-z}~z^{\frac{H}{P^*}-3}~\dd z \to \infty$.

\subsection{Interpretation of expression (\ref{E:condlongR}), emergent risk premium and the equity premium puzzle }

In the non-explosive nonlinear regime $2\alpha \gamma^* > \psi^2$, the discount rate $\gamma_t$ is fluctuating 
around $\gamma^*$, mean-reverting stochastically around it so that $\E (\gamma_t) = \gamma^*$.
A naive investor will deduce that the intrinsic value is then given by $P^*=E/\gamma^*$, which can be argued to 
be the equilibrium price representing the final converged overall consensus on earning yield, and thus
the anticipated long-term return should be $\E(R^*\mid P_0)=\E (P^*/P_0\mid P_0)=P^*/P_0 $. 
Another way to formulate this is to naively surmise from 
expression (\ref{pricing1}) that the price should be characterised by a mean value $\E (P_t)=P^*=E/\gamma^*$.
However, as is well-known, the mean of the inverse is not the inverse of the mean as expression (\ref{E:condlongR})
exemplifies. Indeed, the nonlinear (inverse) dependence of the price $P_t$ on the discount rate $\gamma_t$
implies that $\E (P_t)= \phi P^*$, showing the existence of an amplification factor 
\begin{equation}\label{E:phidef}
	\phi=\left(1-\dfrac{P^*}{H}\right)^{-1}=\left(1-\dfrac{\psi^2 }{2\,\alpha \gamma^*}\right)^{-1}~,~~{\rm for}~  P^*<H~.
\end{equation}
Thus, the expected stationary return is equal to the naively anticipated long-term return multiplied by the factor $\phi$.
This leads naturally to define an ``emergent risk premium rate'' as
\begin{equation}
\rho_e := \ln\phi =-\ln \left(1-{P^*}/{H}\right)~\in (0,\infty)~~{\rm for}~ P^*<H.
\label{tyheyb2gvq}
\end{equation}
To first order in $\dfrac{P^*}{H}$, we have $\rho_e \approx  \dfrac{P^*}{H} = \dfrac{\psi^2 }{2\,\alpha \gamma^*}$.

The volatility amplitude $\psi$ of $\gamma_t$ is the main positive contribution to $\phi$ as seen from 
\begin{align*}
	\dfrac{\partial}{\partial \psi} \ln \phi &=\dfrac{\partial }{\partial \psi}\ln\left(1-\frac{P^*}{H}\right)^{-1}=-\dfrac{\partial }{\partial \psi}\ln\left(1-\dfrac{\psi^2 P^*}{2\alpha E}\right)= \dfrac{\phi \Psi}{\alpha \gamma^*} >0 ~.
\end{align*}  
Figure \ref{F:riskpermium_bubble} shows the dependence of the cumulative return $\E(R_t\mid P_0) = \E(P_t\mid P_0) / P_0$
(equal to the normalised price) as a function of time for several values of $\psi$, as obtained from Equation (\ref{E:condret}).
The larger $\psi$ is, the steeper is the price increase at early times, the larger is the price acceleration at early times (see below 
our discussion of the super-exponential transient regime) and the larger the asymptotic price at long times.
Thus, the larger the amplitude $\psi$ of the fluctuations of $\gamma_t$, the larger is the emergent risk premium rate.

\begin{figure}
	\includegraphics[width=15cm]{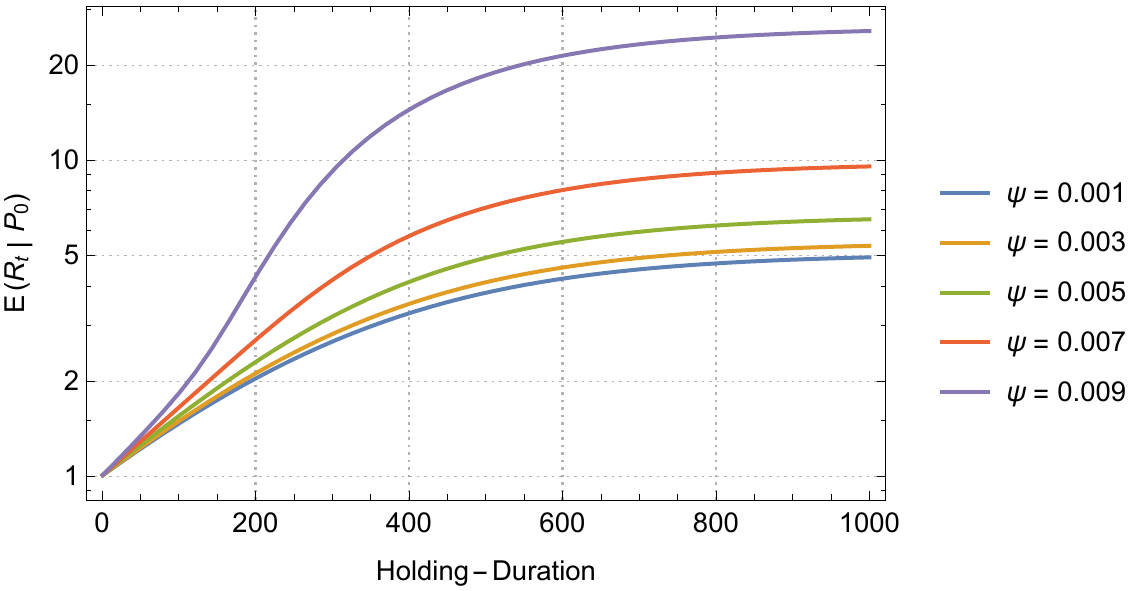}
\caption{Cumulative expected return $\E(R_t\mid P_0) = \E(P_t\mid P_0) / P_0$ as a function of the holding duration for different values of the amplitude $\psi$ 
of the volatility of the $\gamma_t$ process.  All trajectories are generated by setting $E=0.1$, $P^*=10$, $P_0=2$, and $\alpha=0.005$ 
but with five different $\psi$ from $0.001$ to $0.009$. The smallest values of $\psi$ is chosen to let $\psi\sqrt{P/E}$ 
correspond to the typical volatility $0.5\%$ in a standard financial market, considering the typical value of the P/E ratio $=25$.
Correspondingly, $H={2\alpha E}/{\psi^2}$ varies from $12.3$ to $1000$ and $\E(R_{\infty}\mid P_0)= \phi P^*/P_0$ takes values from $5.05$ to $26.84$. }
\label{F:riskpermium_bubble}
\end{figure}    

The fact that $\gamma_t$ is not fixed but fluctuates around $\gamma^*$ leads to an increased long-term return equal to $\ln\phi$
for the investor. This mechanism provides a possible concise description of the
equity premium puzzle \citep{MehraPres85,MehraPres03,Kocherla96}, 
namely the observation that the compensation for holding risky stocks is larger than predicted by
a large class of economic models. Indeed, our simple framework suggests that models that assume constant discount rates 
tend to underestimate future price growth compared with models that account for fluctuating mean-reverting discount rates.

Note that the mechanism for the increased future price growth can be intuitively understood
as roughly resulting from combining different
discounting factors as in (\ref{trhwbgvq}) and (\ref{trhwbgwt2vq}). Indeed, taking the expectation of expression (\ref{trhwbgwt2vq})
yields $\E(P) = {E \over 1- \E(x)}$, assuming that the $x_n$'s are i.i.d. random variables.
From the definition $x := {1 \over 1+\gamma}$ (see (\ref{trhwbgwt2vq})), $x$ is a convex function of $\gamma \in [0, 1]$
and, by Jensen inequality, ${E \over 1- \E(x)}  <  \E\left({E \over 1- x} \right)$.
The gap $\E\left({E \over 1- x} \right) -  {E \over 1- \E(x)}$ is the analog of $\E(P_t) - P^* = (\phi-1)P^*$.
The approximate correspondence between model (\ref{pricing1}) with (\ref{CIRgamma}) 
and model (\ref{trhwbgwt2vq}) is further reinforced by noting that both lead to a power law tail 
for the distribution of prices and returns. For model (\ref{pricing1}), this is shown by Proposition \ref{prop:stableden}
and expression (\ref{supercond}). For model (\ref{trhwbgwt2vq}), this derives from the fact that it is 
a special case \citep{deCalanetal85} of the general of class of Kesten processes with multiplicative and additive stochastic components
\citep{Kesten73,MalSorkes01,LuxSor02,Buraczewski16}.

\subsection{Transient super-exponential price dynamics}

Even in the non-explosive nonlinear regime $2\alpha \gamma^* > \psi^2$ ($P^* < H$), the price exhibits transient
growth regimes that are faster than exponential. A first visual indication is provided in figure \ref{F:riskpermium_bubble},
where the normalised expected price trajectory, plotted with a logarithmic scale as
a function of linear time,  exhibits a clearly visible transient convexity up to $t \approx 200$ for $\psi = 0.009$.
Recall that, in a linear-log plot (linear-in-time with logarithmic-in-price scales), an exponential growth is qualified as a straight line.
A convex curve in such linear-log plot is growing faster than exponential with the instantaneous growth rate (the local tangent to the curve)
growing itself. We are interested in this pattern because it has been previously argued to be a characteristic feature of a financial bubble
\citep{JohSorBr10,SornetteCauwelsbub2015,sorwhy17,SchSor20,Ardila-AlvarezSor21}, since it can only be transient
and has to correct to a long term exponential growth with rate dictated by long-term economic growth.

The existence of transient strongly increasing price trajectories reminiscent of financial bubbles can be observed
in figure \ref{F: price_sim_facet}, which shows simulated price trajectories for the same parameters as 
in figure \ref{F:riskpermium_bubble}. For each value of $\psi$, 10 simulated trajectories are displayed with a specific colour.
For better visibility, the top panel depicts the 20 price trajectories for the two largest $\psi$ values, while the bottom panel
corresponds to the three smallest values. With again the price plotted along the y-axis in logarithmic scale, one can observe 
transient growth regimes that are reminiscent of the prices observed during financial bubbles with super-exponential growth.

\begin{figure}
	\includegraphics[width=15cm]{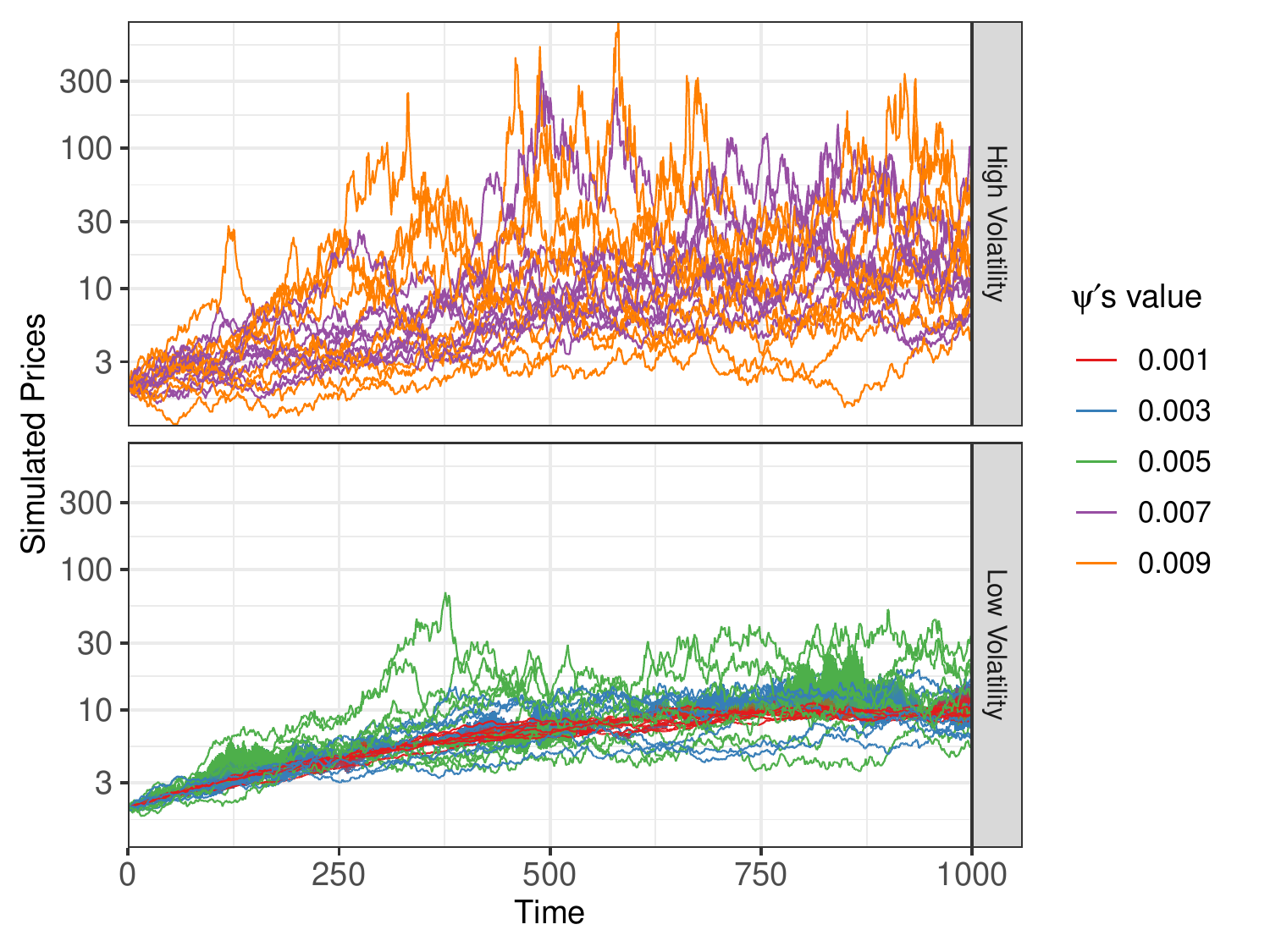}
\caption{Simulated  price trajectories for the same parameters as in figure \ref{F:riskpermium_bubble}  and for the five values of $\psi$
given in the inset.  For each value of $\psi$, 10 simulated trajectories are displayed with a specific colour indicated in the inset. The 20 price
trajectories corresponding to the two largest $\psi$ values are shown in the top panel, and the 30 price trajectories for the three lowest
$\psi$ values are in the lower panel. Note the logarithmic scale used for the y-axis.}
\label{F: price_sim_facet}
\end{figure}    

Figure \ref{F:super_price} is similar to figure \ref{F:riskpermium_bubble} showing  $\E(P_t\mid P_0) / P_0$
as a function of time but for different initial price $P_0$. Although the five average price trajectories
converge to the same long-term expected price level $\phi P^*\approx 53.5$, they first exhibit a convex growth shape
in the linear-log plot, corresponding to a super-exponential transient regime, as previously documented.
This convex (super-exponential) growth then transitions to a concave growth, associated with a slowdown of the price
increase until convergence to the steady-state level. The smaller the initial price $P_0$, the longer is the duration
of the super-exponential regime.
 
\begin{figure}
	\includegraphics[width=15cm]{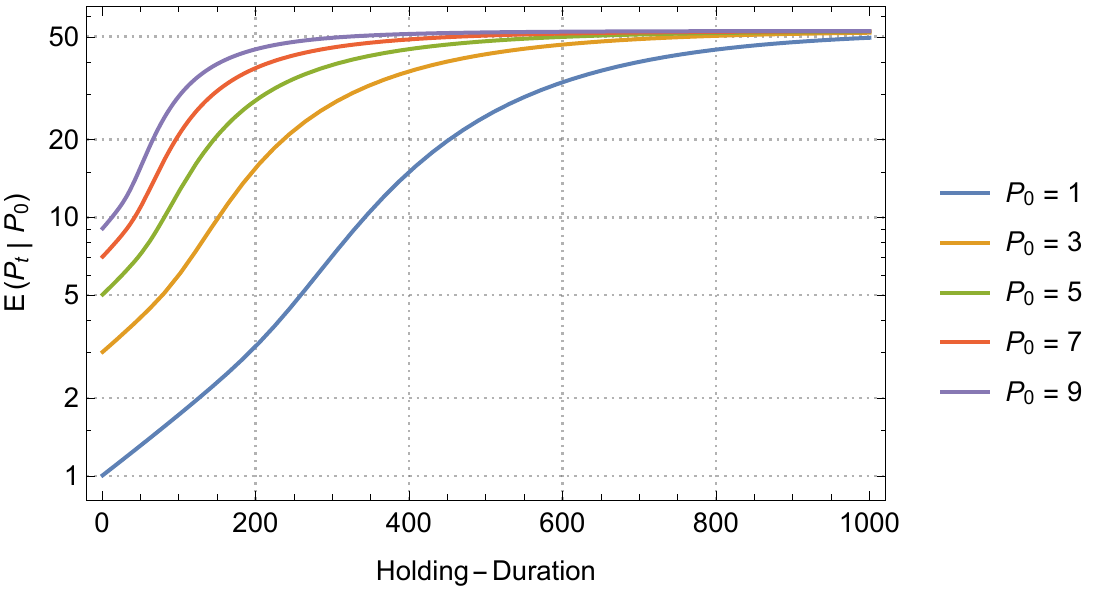}
\caption{The evolution of expected price as a function of the holding duration in a bubble process.
The parameters are $E=0.1, P^*=10, \alpha=0.005$ with $\psi=0.009$ corresponding to the largest values used in 
figures \ref{F:riskpermium_bubble} and \ref{F: price_sim_facet}. }
\label{F:super_price}
\end{figure}  

The following proposition specifies the conditions needed for the existence of the super-exponential growth regime
and characterises its duration. 

\begin{thm}\label{prop:super}
Let the price process be characterised by the SDE (\ref{E:Pb}) in the non-explosive nonlinear regime $2\alpha \gamma^* > \psi^2$ ($P^* < H$).
Given an initial price $P_0$, 
and with the definitions $q= \frac{H}{P^*}-1$ (\ref{trhtgbq4}), $H=\dfrac{2\alpha E}{\psi^2}$ and $P^* = E / \gamma^*$, $g(\alpha)=e^{\alpha}-1$, and
$
	\Omega=\dfrac{e^{\frac{H}{P_0\,g(\alpha) }}\left(-\frac{H}{P_0\,g(\alpha)}\right)^q}{\Gamma(q)-\Gamma(q,-\frac{H}{P_0\,g(\alpha)})}~,
$
under the condition that
\begin{align}\label{supercond}
	-{g(\alpha)\,e^{\alpha}}P_0\,\Omega^2+{[\,e^{\alpha}H+g(\alpha)(1+e^{\alpha}q)P_0\,]}\,\Omega-{(e^{\alpha}+1)H+P_0\,g(\alpha)(q-1)}>0~,
\end{align}
then the expected price persists in super-exponential growth from $t=0$ till at least $t=1$ and its duration $t_c$ is determined by finding the time 
that is solution of the following equation:
\begin{align}\label{solvingsupert}
	\dfrac{\partial^2}{\partial t^2}\left[\,-\dfrac{H}{P_0}\dfrac{1}{e^{\alpha t}-1}-\ln(1-e^{-\alpha t})+\ln\dfrac{H}{q}+\ln{}_1F_1\left(q,q+1,-\dfrac{H}{P_0}\dfrac{1}{e^{\alpha t}-1}\right)\,\right]=0. 
\end{align} 
The choice of the minimum duration of $1$ is arbitrary and is readily generalisable to any positive value.
\end{thm}

\begin{proof}
The proof of this proposition is given in Appendix \ref{A:propsuper}. {\hfill $\square$}
\end{proof}

\begin{figure}
	\includegraphics[width=15cm]{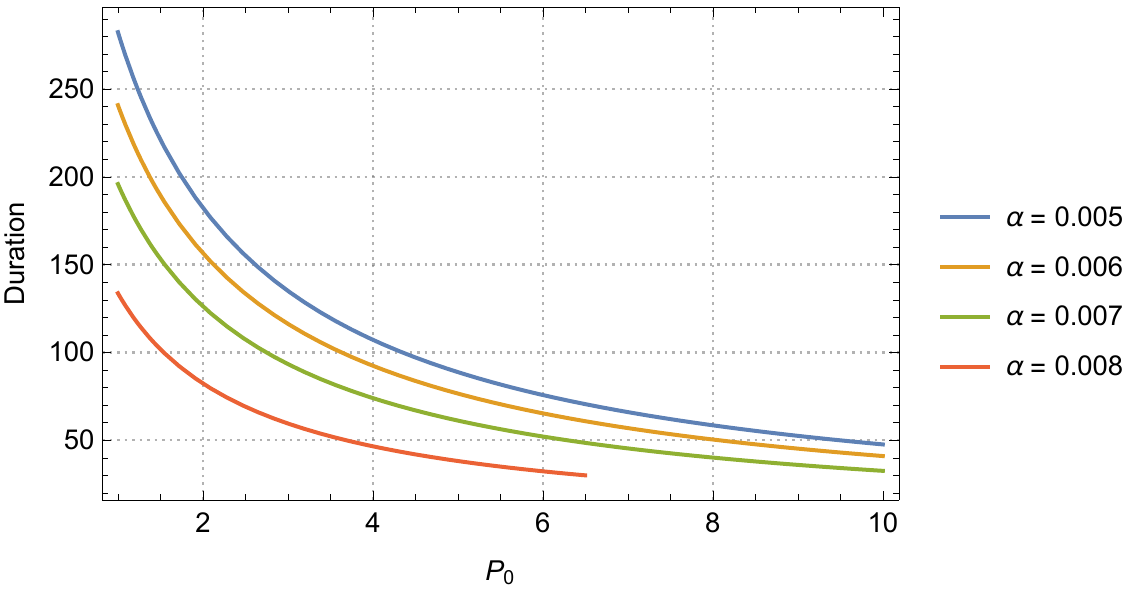}
\caption{Dependence of the duration $t_c$ of the super-exponential growth regime obtained by solving equation (\ref{solvingsupert})
as a function of the initial price $P_0$. The parameters are $E=0.1$, $\psi=0.009$, $P^*=10$. The mean-reverting strength $\alpha$
of the yield process $\gamma_t$ takes values from $0.005$ to $0.008$ as indicated in the inset, with their colour codes.
We scan values of $\alpha$ such that condition (\ref{supercond}) is valid, ensuring a minimum duration of the super-exponential regime.
This corresponds to having $\alpha$ approximately $\in [0.005; 0.008]$. For instance, taking 
$\alpha=0.004$ leads to $H=9.87<10=P^*$. Similarly, for $\alpha>0.008$, condition (\ref{supercond}) is violated.}
\label{F:super_t}
\end{figure}  

Figure \ref{F:super_t} shows the dependence of the duration $t_c$ of the super-exponential regime obtained by solving equation (\ref{solvingsupert})
as a function of the initial price $P_0$ at time $t=0$. The smaller $P_0$ and the smaller $\alpha$, the longer is the duration
of the super-exponential regime. The curve shown in figure \ref{F:super_t} for $\alpha=0.005$ is consistent with the results of figure \ref{F:super_price}.    

Figure \ref{F:super_t2} shows the dependence of the duration $t_c$ of the super-exponential regime 
as a function of the volatility amplitude $\psi$.  It exhibits
a concave shape with maximum value obtained for $\dfrac{\psi^2}{2\alpha \gamma^*}\approx 0.7$. The results can be explained as follows. When $H$ is too small to approach $P^*$ ($P^*/H \to 1$), the long-run expected price $\phi P^*$ becomes infinite such that it leaves enough time for the expected price to converge. Thus, the duration of super-exponential growth can be largely cut. However, when $H$ is sufficiently large ($P^*/H \to 1/2$), it requires quite a tight time for the expected price to converge and leave little time for the duration of super-exponential growth.


\begin{figure}
	\includegraphics[width=15cm]{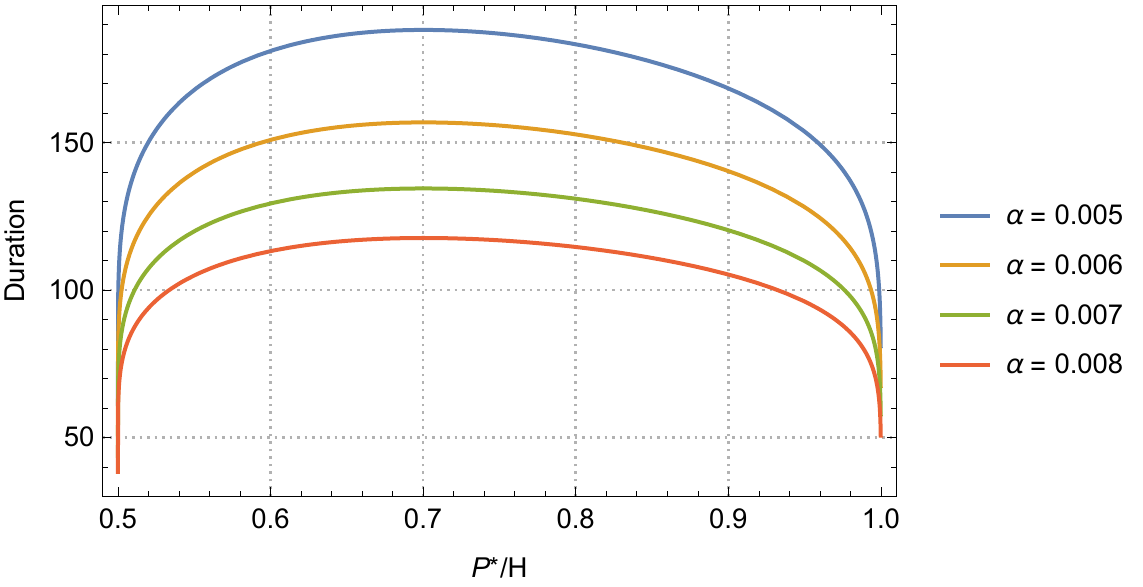}
\caption{Duration of the transient super-exponential growth regime obtained by solving equation (\ref{solvingsupert})
as a function of $\dfrac{P^*}{H}:=\dfrac{\psi^2}{2\alpha \gamma^*}$. The parameters are $E=0.1$, $P_0=2$, $P^*=10$.
Each curve corresponds to a fixed value of $\alpha$ from $0.005$ to $0.008$ as shown in the inset.}
\label{F:super_t2}
\end{figure}

\section{ Calibration of the model to real data}

\subsection{Quasi-maximum likelihood calibration}

The price process (\ref{E:Pb}), which derives from expressions (\ref{pricing1}) with (\ref{E:bubbleCIRgamma}) in Proposition \ref{proposition1},
can be calibrated by transforming the empirical price back into its corresponding earning yield (or discount rate) $\gamma_t$, which is then calibrated to 
the CIR process. For a more stable numerical implementation, instead of the original form \eqref{CIRgamma} or (\ref{E:bubbleCIRgamma}), 
it is preferable to use the following form
\begin{equation}\label{E: CIRexp}
	\dd \gamma_t=(b-\alpha\, \gamma_t)\,\dd t+\psi \sqrt{\gamma_t}\,\dd W_t~.
\end{equation}
After the statistical estimates $\widehat{\alpha}$ and $\widehat{b}$ are obtained, $\widehat{\gamma^*}$ can be recovered from $\widehat{b}/\widehat{\alpha}$. 
In principle, one could implement an exact likelihood inference, since the CIR process has an explicit analytical conditional density \eqref{E:CIRdensity0},
i.e. a non-central $\chi^2$ distribution. However, the involved modified Bessel function $I_q(\cdot)$ makes it difficult to 
calculate the conditional distribution, as it tends to explode, inducing numerical implementation overflows. A
solution is to implement the quasi-maximum likelihood approach to take advantage of the Euler discretisation of Equation (\ref{E: CIRexp}),
\begin{equation*}
	\gamma_{t+\Delta_s}-\gamma_t=(b-\alpha\, \gamma_t)\,\Delta_s+ (\psi^2\gamma_t)^{\frac{1}{2}}(W_{t+\Delta_s}-W_t)~,
\end{equation*}
where $\Delta_s$ denotes the time lag between two consecutive observations of the process. As $(W_{t+\Delta_s}-W_t\sim \mathcal{N}(0, \Delta_s)$, the approximate transition density is
\begin{equation*}
	\gamma_{t+\Delta_s}\mid \gamma_t \sim \mathcal{N}( \gamma_t+(b-\alpha\, \gamma_t)\,\Delta_s, ~ (\psi^2\gamma_t)^{\frac{1}{2}} \Delta_s)~.
\end{equation*}
Thus, the quasi-log-likelihood for a series  $\pmb{\gamma}:=\{\gamma_t\}_{t=0,1,\cdots,n}$ of observed data determined for a specific set of  parameters is (neglecting a constant term)
\begin{equation}
\begin{split}
	\ln\mathcal{L}(\pmb{\gamma},(b,\alpha,\psi))&=-\dfrac{1}{2}\sum_{t=1}^{t=n}\left\{\ln \psi^2\gamma_{t-1}+\dfrac{1}{\Delta^2_s \psi^2\gamma_{t-1}}\cdot [\Delta \gamma_t-\Delta_s (b-\alpha \gamma_{t-1})]^2 \right\},
\end{split}
\end{equation}
where $\Delta \gamma_t=\gamma_t-\gamma_{t-1}$. The optimal parameters are obtained to satisfy
\begin{equation*}
	\max_{\{b,\alpha,\psi\}} \ln\mathcal{L}(\pmb{\gamma},(b,\alpha,\psi))~.
\end{equation*}
Further, the corresponding variance for each estimated parameter can be obtained from the diagonal terms of the following Fisher information matrix,
\begin{equation*}
	\mathbb{V}ar (\widehat{\pmb{\theta}})=\left(\dfrac{\partial^2\mathcal{L}}{
\partial\pmb{\theta}\partial\pmb{\theta}}\right)^{-1}\biggl|_{\pmb{\theta}=\widehat{\pmb{\theta}}},\qquad \pmb{\theta}=(b,\alpha,\psi)~.
\end{equation*} 

Given a time series of daily prices $\{P_t\}_{t=1,\cdots,n}$ with $n$ observations, we calculate the corresponding series $\{\gamma_t\}_{t=0,1,\cdots,n}$
using $\gamma_t=E/P_t$. Because we focus on empirical time series that have exhibited transient bubble regimes,
we assume that the volatility of $E$ can be neglected and we treat $E$ as a constant in the calibration. We adopt the L-BFGS-B method to optimise the quasi-log-likelihood. The limited-memory Broyden–Fletcher–Goldfarb–Shanno (BFGS) optimisation algorithm belongs to the family of quasi-Newton BFGS algorithms,
which uses a limited amount of computer memory. The L-BFGS-B extends L-BFGS to handle simple bound constraints in search of the optimal parameters \citep{byrd1995}. Moreover, as $\Delta_s\to 0$ is the necessary condition for the quasi-maximum likelihood method to obtain consistent estimations, we use $\Delta_s=1/252$ (data of a year including $252$ observations) in our implementation.

\subsection{Description of the five empirical data sets}

We calibrate the model for the following five well-known historical bubbles. 

\vskip 0.3cm
\noindent
{\bf The US Dotcom bubble ending in mid-2000}. It is often called the Internet-Communication-Technology (ICT) bubble and regarded as a textbook-style bubble (\cite{homm2012,Phillips2011}). Over 1999, the cumulative return of the internet stock index (represented by NASDAQ) surged 
with an eightfold multiplication of the price; in contrast, for non-internet stocks, the return was smaller than 20\% \citep{kaizoji2015}. However, by mid-2001, a year after the bubble ended, the cumulative sum of the astronomical returns for internet stocks had completely evaporated. From April 11, 2000, the NASDAQ index underwent five days of consecutive sharp drops, with a drawdown of 25.8\%. It was the largest drawdown since the index reached the peak of 5132.5 on March 10, 2000. For model calibration, we employ the NASDAQ daily close price as $P_t$ and set April 11, 2000, as the end of the time window, with the corresponding date, April 12, 1999, as the start of the time window. According to the report from the NASDAQ offset market exchange white book, the P/E ratio was approximately $150$ in mid-1999. Thus, we set $\gamma_0=\gamma_{t_{\text{start}}}=1/150$ to be the initial value of $\gamma_t$ in the calibration window. Hereafter, this bubble is referred to as NASDAQ 2000. 
\item {\bf US stock market bubble ending in October 1987}. This bubble is often identified in retrospect as the most striking drop on October 19, 1987, known as the ``Black Monday.'' Much work has been conducted to unravel the origin(s) of the crash \citep{sorwhy17}. However, no clear cause has been singled out. Some commentators have ascribed the crash to the overinflated prices from a speculative bubble during the earlier period, which pushed global stock markets into an unsustainable state. The sharp drop on Black Monday is only the grand finale of consecutive market declines with an impressive cumulative depreciation of $29.6\%$, which began on October 6, 1987, lasting for nearly two trading weeks. For model calibration, we employ the S\&P 500 daily close price as $P_t$ and set October 6, 1987, as the end of the time window, with the corresponding date, October 7, 1986, as the start of the time window. According to the earnings data from S\&P 500 Global, the P/E ratio for S\&P 500 was approximately $6.9$ in October 1986. Thus, we set $\gamma_0=\gamma_{t_{\text{start}}}=1/6.9$ to be the initial value of $\gamma_t$ in the calibration window. Hereafter, this bubble is referred to as S\&P 500 1987.  

\vskip 0.3cm
\noindent

{\bf US bubbles ending in October 1929}. This bubble is famously the herald of the Great Depression. It was preceded by extraordinary growth and prosperity on Wall Street in the late 1920s, with the Dow Jones Industrial Average index (DJIA) peaking at 381 on September 3, 1929. This bubble has been analysed in great detail by \cite{galbraith2009great}. Its development is characterised by many stories on economic prosperity, with the reconstruction boom after the first world war. From October 11, 1929, incredible consecutive drops of DJIA kicked off. It only took one trading month for the price of DJIA to fall to $195.4$ from $392.4$, with a huge drawdown of $50\%$. For model calibration, we employ the DJIA daily close price as $P_t$ and set October 11, 1929, as the end of the time window, with the corresponding date, October 12, 1928, as the start of the time window. As the annual normalised P/E ratio of DJIA in 1928 is $12.5$, we set $\gamma_0=\gamma_{t_{\text{start}}}=1/12.5$ to be the initial value of $\gamma_t$ in the calibration window. Hereafter, this bubble is referred to as DJIA 1929.  

\vskip 0.3cm
\noindent
{\bf Chinese stock market bubble ending at the beginning of 2008}. After the reform of the split-share structure launched in 2006, the most representative stock index for the market of A-shares (i.e., the Shanghai Stock Exchange Composite [SSEC] index) underwent an exuberant growth. The growth was fuelled by compelling growth stories, confirming from all directions great positive outlooks based on the rapid fundamental growth of the Chinese economy, its sky-rocketing export surpluses inducing enormous reserves and liquidity, the strength of its currency, and the coming Olympic Games offering a path to extraordinary expected prosperity \citep{lin2018}. Multiple anecdotal sources strongly posit that Chinese investors changed from prudent to very confident in the future prospects and gains. In 2007, the SSEC peaked at 6124 in October from 2600 in January (corresponding to a relative appreciation of 235.5\%) within 10 months. However, after an anodyne rally in the first two weeks of 2018, SSEC began to slump on January 14, 2008. The following eight months saw the index return to the original $2600$ level, capping a breathtaking roller-coaster performance for the bubble. For model calibration, we employ SSEC daily close price as $P_t$ and set January 14, 2008, as the end of the time window, with the corresponding date, January 15, 2007, as the start of the time window. According to the Census and Economic Information Centre (CEIC) data, the P/E ratio for SSEC was approximately $20$ in January 1999; we set $\gamma_0=\gamma_{t_{\text{start}}}=1/20$ to be the initial value of $\gamma_t$ in calibration window. Hereafter, this bubble is referred to as SSEC 2008.

\vskip 0.3cm
\noindent
{\bf Chinese stock market bubble ending in the mid of 2015}.
This bubble started around mid-2014, inducing an approximate growth of 150\% in a year, crashed thereafter from mid-2015. The origin of this bubble is often considered to be quite different from SSEC 2008. It is thought to result from investors utilising strong leverage that led to an amplification and disconnection between the price and the realities of economic activity and corporate earnings. Indeed, the Chinese government encouraged small retail investors to join in investing in the stock market, and around 7\% of China's population did profit from the easy access to credit for investment purposes \citep{sornette2015chinese}. However, as a remarkable feature, the price's roller-coaster performance of this bubble is quite similar to SSEC 2008. For model calibration, we employ SSEC daily close price as $P_t$ and set June 30, 2008, (which opens a consecutive seven-day drop, the longest losing streak since the peak of the bubble) as the end of the time window, with the corresponding date, July 1, 2014, as the start of the time window. According to the CEIC data, the P/E ratio for SSEC was approximately $10$ in July 2014. Thus, we set $\gamma_0=\gamma_{t_{\text{start}}}=1/10$ to be the initial value of $\gamma_t$ in calibration window. Hereafter, this bubble is referred to as SSEC 2015.

\begin{figure}
\centering
\includegraphics[width=17cm]{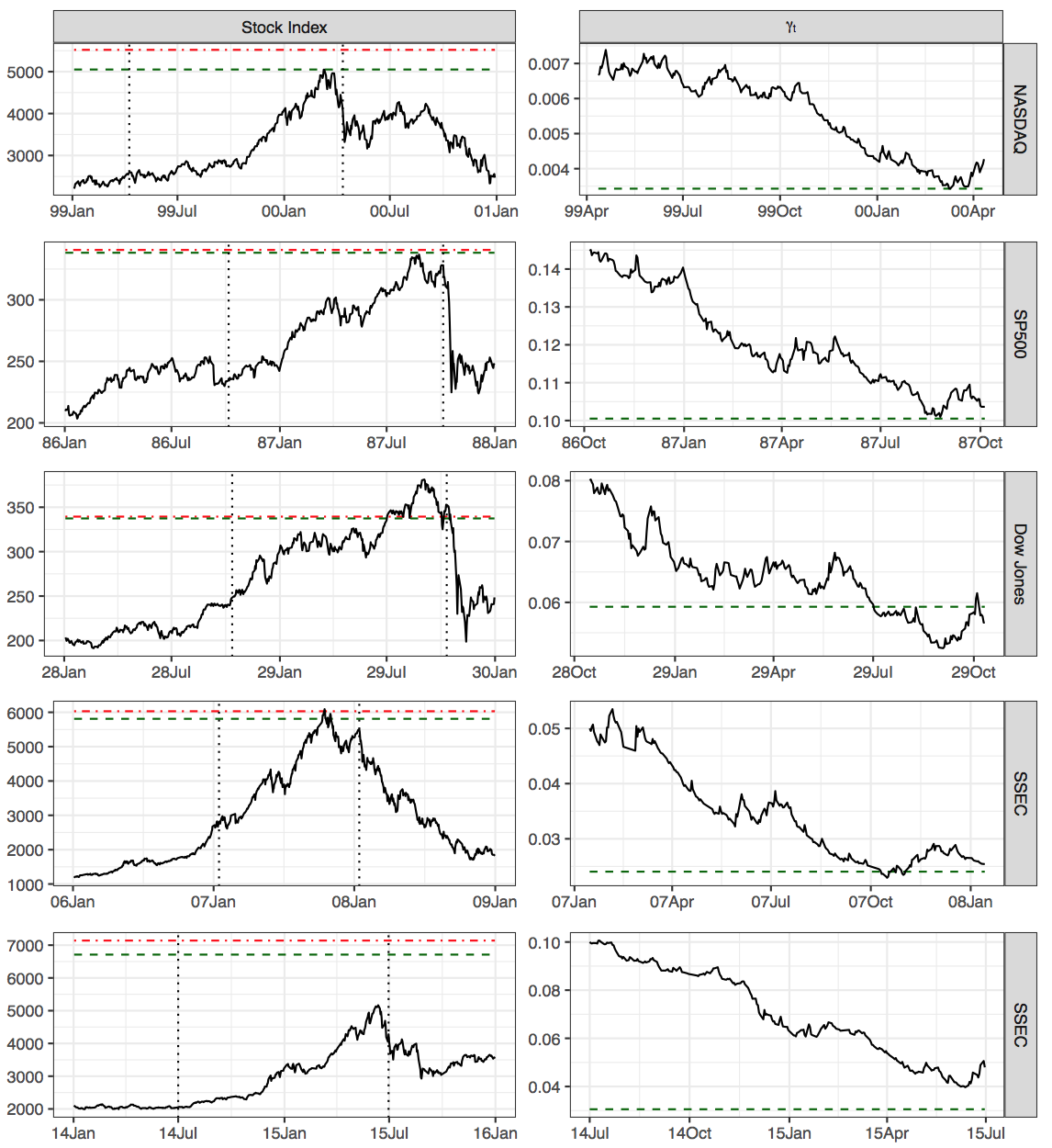}
\caption{Calibration of the model for five historical bubbles. The left panels show the stock index prices for each case. 
The calibration window for each case is delineated by the two dotted vertical lines. 
For each bubble, the estimation algorithm is implemented in a time window of one year 
ending just before the crash following the bubble.  
The right panels present the corresponding time series of $\gamma_t$ in the calibration window of one year ending just before each crash.  
The horizontal green dashed lines indicate the estimated $\widehat{P^*}=E/\widehat{\gamma^*}$ and $\widehat{\gamma^*}$ respectively in the left and right panels. The red dotted-dashed line represents the estimated $\phi P^*$.}\label{F:calibration}
\end{figure}

\subsection{Results of the calibration on the five empirical bubble examples}

For each time series of $\gamma_t$, we choose $\{b=0,\alpha=0,\psi=0\}$ as the lower bound and the sufficiently large values $\{b=100,\alpha=100,\psi=100\}$ as the upper bound of the parameter space to implement the L-BFGS-B algorithm. The initial values for optimisation are all set to be $\{b=0.01,\alpha=0.01,\psi=0.01\}$ for the above five bubble cases. After $\widehat{b},\widehat{\alpha},\widehat{\psi}$ has been obtained, we derive $\widehat{\gamma^*}=\widehat{b}/\widehat{\alpha}$. Further, $\widehat{P^*}=E/\widehat{\gamma^*}$ , $\widehat{\phi}=[1-{\widehat{\psi}^2 }/{(2\,\widehat{\alpha} \widehat{\gamma^*}})]^{-1}$, and $\widehat{P^{\dagger}}=\widehat{\phi}\widehat{P^*}$ can be sequentially calculated. Table \ref{T:1} lists all calibrated parameters, where $t_{\text{start}}$ and $t_{\text{end}}$ respectively gives the start and end of date for the calibration time window, and $P_{\text{max}}$ gives the maximum price of the stock index within the calibration time window. $\ln\mathcal{L}$ denotes the quasi-maximum log-likelihood for the empirical $\gamma_t$ over the time window based on the estimated parameters.  

\renewcommand\arraystretch{1.8}
\begin{table}[!p]
\begin{tabular}{cccccc}
\toprule
                        & NASDAQ 2000                                                                           & S\&P500 1987                                                                          & Dow Jones 1929                                                            & SSEC 2008                                                                           & SSEC 2015                                                                           \\ \midrule
$t_{\text{start}}$      & 1999/04/12                                                                            & 1986/10/06                                                                            & 1928/10/12                                                                & 2007/01/15                                                                          & 2014/07/01                                                                          \\
$t_{\text{end}}$        & 2000/04/11                                                                            & 1987/10/07                                                                            & 1929/10/11                                                                & 2008/01/14                                                                          & 2015/06/30                                                                          \\
$P_{\text{max}}$        & 5132.5                                                                                & 337.9                                                                                 & 381.1                                                                     & 6092.1                                                                              & 5166.35                                                                             \\
$\widehat{b}$           &\renewcommand\arraystretch{0.5} \begin{tabular}[c]{@{}c@{}}$1.51\times 10^{-5}$\\ ($3.23\times 10^{-5}$)\end{tabular} & \renewcommand\arraystretch{0.5}\begin{tabular}[c]{@{}c@{}}$8.17\times 10^{-4}$\\ ($1.24\times 10^{-4}$)\end{tabular} &\renewcommand\arraystretch{0.5} \begin{tabular}[c]{@{}c@{}}$0.0012$\\ ($5.89\times 10^{-4}$)\end{tabular} &\renewcommand\arraystretch{0.5} \begin{tabular}[c]{@{}c@{}}$2.38\times 10^{-4}$\\ $(2.04\times 10^{-4})$\end{tabular} &\renewcommand\arraystretch{0.5} \begin{tabular}[c]{@{}c@{}}$1.63\times 10^{-4}$\\ $(2.65\times 10^{-4})$\end{tabular} \\[1.5ex]
$\widehat{\alpha}$      &\renewcommand\arraystretch{0.5} \begin{tabular}[c]{@{}c@{}}$0.0044$\\ ($0.0059$)\end{tabular}                         & \renewcommand\arraystretch{0.5}\begin{tabular}[c]{@{}c@{}}$0.0081$\\ ($0.0086$)\end{tabular}                         &\renewcommand\arraystretch{0.5} \renewcommand\arraystretch{0.5}\begin{tabular}[c]{@{}c@{}}$0.0194$\\ ($0.0092$)\end{tabular}             & \renewcommand\arraystretch{0.5}\begin{tabular}[c]{@{}c@{}}$0.0099$\\ ($0.0061$)\end{tabular}                       & \renewcommand\arraystretch{0.5}\begin{tabular}[c]{@{}c@{}}$0.0054$\\ ($0.0039$)\end{tabular}                       \\
$\widehat{\gamma^{*}}$  & $0.034$                                                                               & $0.101$                                                                               & $0.059$                                                                   & $0.024$                                                                             & $0.031$                                                                             \\
$\widehat{\psi}$        &\renewcommand\arraystretch{0.5} \begin{tabular}[c]{@{}c@{}}$0.0016$\\ ($8.81\times 10^{-8}$)\end{tabular}             & \renewcommand\arraystretch{0.5}\begin{tabular}[c]{@{}c@{}}$0.0033$\\ ($3.49\times 10^{-7}$)\end{tabular}             & \renewcommand\arraystretch{0.5}\begin{tabular}[c]{@{}c@{}}$0.0037$\\ ($4.48\times 10^{-7}$)\end{tabular} & \renewcommand\arraystretch{0.5}\begin{tabular}[c]{@{}c@{}}$0.0042$\\ ($5.86\times 10^{-7}$)\end{tabular}           &\renewcommand\arraystretch{0.5} \begin{tabular}[c]{@{}c@{}}$0.0044$\\ ($6.70\times 10^{-7}$)\end{tabular}           \\
$\widehat{\phi}$        & $1.093$                                                                               & $1.007$                                                                               & $1.006$                                                                   & $1.037$                                                                             & $1.064$                                                                             \\
$\widehat{H}$           & $5.94\times 10^5$                                                                     & $5.04\times 10^5$                                                                     & $5.61\times 10^5$                                                         & $1.59\times 10^6$                                                                   & $1.12\times 10^6$                                                                   \\
$\widehat{P^*}$         & $5051.3$                                                                              & $338.2$                                                                               & $337.4$                                                                   & $5810.4$                                                                            & $6711.6$                                                                            \\
$\widehat{P^{\dagger}}$ & $5520.6$                                                                              & $340.5$                                                                               & $339.4$                                                                   & $6030.7$                                                                            & $7140.1$                                                                            \\
$2 \ln \mathcal{L}$    & $-3916.5$                                                                             & $-2706.0$                                                                             & $-2773.0$                                                                 & $-2809.0$                                                                           & $-2620.0$\\
\bottomrule                                                                         
\end{tabular}
\caption{Parameters of model (\ref{E:Pb}), which derives from expressions (\ref{pricing1}) with (\ref{E:bubbleCIRgamma}) in Proposition \ref{proposition1},
obtained by calibrating five historical bubbles via quasi-maximum likelihood estimation method. The standard error calculated from Fisher information matrix are provided in the parentheses.
}\label{T:1}
\end{table}

As shown in table  \ref{T:1}, $\widehat{\alpha}$ and $\widehat{\psi}$ are respectively in the typical range of $0.005\sim 0.02$ and $0.002\sim 0.004$ for all bubbles. Moreover, the estimated $\widehat{P^*}$ are not far from the peak price $P_{\text{max}}$, except for SSEC 2015. This suggests that, before the final burst, the first four bubbles had almost reached maturation. SSEC 2015 seems to collapse too early before the instantaneous collective market belief on earning yield $\gamma_t$
had time to reach its equilibrium state corresponding to $\widehat{\gamma^*}$, leading $P_{\text{max}}$ to remain
much smaller than $P^*$. This exception could be rationalised by the evidence that the Chinese real estate market and the overall economy synchronously cooled significantly when this bubble developed. As the real estate sector largely drove the Chinese economy in the first five years of the 2010s, the cooling of the real estate market could have dragged down the expected growth rate of the whole Chinese economy. It acted as an external pressure to cause a regime shift for $\gamma^*$ in mid-2015, stopping the bubble early from further growth. 

The left panels of Figure \ref{F:calibration} show the price trajectories for each of the five stock indices studied here, in time windows showing clearly the 
ascending bubble regimes followed by the crashes and following drawdowns. The right panels show the 
corresponding earning yields $\gamma_t$. The calibrated $\widehat{P^*}$ are found close to the peak of the prices, except for SSEC 2015. 
It is interesting to note that $\widehat{P^*}$ and $\widehat{P^{\dagger}}=\widehat{\phi}\widehat{P^*}$ are almost identical for S\&P500 1987 and Dow Jones 1929, indicating that   the emergent risk premium $\ln\phi$ for these bubbles is rather small. As table \ref{T:1} shows, this might be ascribed to different reasons. For S\&P500 1987, the small $\phi$ came from a large $\gamma^*$ (recall Equation (\ref{E:phidef})) that reflects a large expected equilibrium yield in that bubble. For Dow Jones 1929, 
the small $\phi$ results from the large value of $\alpha$, which controls the typical time scale for $\gamma_t$ to converge to $\gamma^*$.

\subsection{Model comparison to test to goodness of fit of the price process (\ref{E:Pb})}

All the above calibrations and corresponding analyses have been performed under the assumption that the obtained yield process
$\gamma_t$ truly obeys the CIR process in the form (\ref{E: CIRexp}).  Recall that $\gamma_t$ is obtained from the 
obtained price process by using (\ref{pricing1}). Indeed, if the price follows equation (\ref{E:Pb}),
then $\gamma_t$ follows equation (\ref{E:bubbleCIRgamma}) and vice-versa, according to Proposition \ref{proposition1}.

It remains to be shown empirically whether the dynamics characterised by eq.(\ref{E: CIRexp}) for $\gamma_t$ fits data better 
than alternative yield processes. We now turn to this question by developing tests based on 
$\phi$-divergence statistics specifically adapted to models prescribed in terms of stochastic differential equations (SDE) \citep{iacus2009,pardo2018}. In a nutshell, this approach tests whether there is a statistically significant difference between an alternative model compared with the model of reference (null hypothesis) 
in the explanatory power of a given dataset, where  the difference between the two (approximated) parametric models is measured in terms of certain $\phi$-divergence measures.  

Let us denote the parameter vector $\pmb{\theta}=(\pmb{\theta}',\pmb{\theta}^{''})$ that contains all parameters  in the following generic SDE describing an ergodic diffusion process $X_t$
\begin{equation}
	\dd X_t=\mu(\pmb{\theta}', X_t)\,\dd t + \sigma(\pmb{\theta}^{''}, X_t)\,\dd W_t, 
\end{equation}
where parameter vector $\pmb{\theta}'$ and $\pmb{\theta}^{''}$ are respectively with dimension $p$ and $q$.  Let $f(X, \pmb{\theta})$ be the probability density for the observations $X$. 
Let us assume that the model under testing is calibrated and obtains parameters $\pmb{\theta}_0$, while the alternative model is calibrated and obtains parameters $\widehat{\pmb{\theta}}$. 
The $\phi$-divergence measure between the model $f(X, \pmb{\theta}_0)$ (null hypothesis) and $f(X, \widehat{\pmb{\theta}})$  is defined as the following expected function of likelihood ratios, 
\begin{equation}\label{E:phidev}
	D_{\phi} (\widehat{\pmb{\theta}},\pmb{\theta}_0)=\E _{\pmb{\theta}_0}\left[\,\phi\left(\dfrac{f(X,\widehat{\pmb{\theta}})}{f(X, \pmb{\theta}_0)}\right)\,\right]~,
\end{equation}
where $\phi(x)$ is a convex continuous function and is restricted on $[0,+\infty]$ with $\phi(1)=\phi'(1)=0$. Using different $\phi(x)$ allows the above divergence 
to be translated into a large body of well-known distance measures for two probability densities. For example, 
\begin{itemize}
	\item  for $\phi(x)=x (\ln x -1) +1$, $D_{\phi} $ recovers to the Kullback-Leibler (KL) divergence \citep{KLdivergence}  $D_{KL}(f(X, \widehat{\pmb{\theta}})||f(X, \pmb{\theta}_0))$;
\item for $\phi(x)=\left(\dfrac{x-1}{x+1}\right)^2$, $D_{\phi}$ recovers to the Balakrishnan-Sanghvi  (BS) divergence \citep{balakrishnan1968} equals to\\ $\displaystyle\int \left(\dfrac{f(X, \widehat{\pmb{\theta}})-f(X, \pmb{\theta}_0)}{f(X, \widehat{\pmb{\theta}})+f(X, \pmb{\theta}_0)}\right)^2f(X, \pmb{\theta}_0)\,\dd X$;
\item for $\phi(x)=-2 x^{\frac{1}{2}}+2x+1$, $D_{\phi}$ recovers to the Rathie-Kannappan (RK) divergence \citep{RKdivergence}, which is propotional to \\$\displaystyle \int \left[1-\left(\dfrac{f(X, \widehat{\pmb{\theta}})}{f(X, \pmb{\theta}_0)}\right) ^{-\frac{1}{2}} \right] f(X, \pmb{\theta}_0)\,\dd X$. 
\end{itemize}

Given a sample of $n$ observations and some asymptotically efficient estimator, such as the quasi-log-likelihood estimator $\widehat{\pmb{\theta}}$,  \cite{iacus2013} has proven that the $\phi$-divergence statistic $T_{\phi,n}=2n D_{\phi,n}^{\text{EXP}}(\widehat{\pmb{\theta}},\pmb{\theta}_0)$ is asymptotically distribution-free under $H_0: \pmb{\theta}=\pmb{\theta}_0$ versus   $H_1: \pmb{\theta}\neq\pmb{\theta}_0$,  and
\begin{equation*} 
	T_{\phi,n}\overset{d}{\longrightarrow} \chi^2_{p+q}, \quad \text{when}~ n\to \infty, \Delta_s\to 0.
\end{equation*}

$D_{\phi,n}^{\text{EXP}}(\widehat{\pmb{\theta}},\pmb{\theta}_0)$ is the empirical version for the theoretical $\phi$-divergence measure (\ref{E:phidev}), which reads, 
\begin{equation*}
	D_{\phi,n}^{\text{EXP}}(\widehat{\pmb{\theta}},\pmb{\theta}_0)=\dfrac{1}{n}\sum_{t=1}^{t=n} \phi\left(\dfrac{p_t(\widehat{\pmb{\theta}})}{p_t(\pmb{\theta}_0)}\right)~,
\end{equation*}
where $p_t(\pmb{\theta})=\exp\left(\dfrac{1}{2}\left\{\ln \sigma^2(\pmb{\theta}^{''}, X_{t-1})+\dfrac{1}{\Delta_s^2 \sigma^2(\pmb{\theta}^{''}, X_{t-1})}\cdot[\Delta X_t-\Delta_s \mu(\pmb{\theta}^{'}, X_{t-1})]^2\right\}\right)$.

\vspace{3mm}

To implement the test for a given time series $\{\gamma_t\}_{t=1}^{t=n}$, we first estimate $\widehat{\pmb{\theta}}$ for the chosen alternative model 
by maximizing its corresponding quasi-log-likelihood. The second step is to calculate the above $\phi$-divergence statistic  $T_{\phi,n}$ by taking into account $\pmb{\theta}_0$, the parameters for the model under testing. Then, given the significance level $\alpha$,  we test whether to reject $H_0$ by examining if $T_{\phi,n} > c_{\alpha}$ where $c_{\alpha}$ is the $1-\alpha$ quantile of the limiting random variable that follows the $\chi^2_{p+q}$ distribution.  

For the $\gamma_t$ time series obtained for each of the five historical bubbles studied above, we compare the CIR model to the following three alternative models:
\begin{itemize}
	\item Brownian Motion (BM): $\dd \gamma_t = b\, \dd t + \psi\, \dd W_t$, 
 	\item  Geometric Brownian Motion: (GBM): $\dd \gamma_t = -\alpha \gamma_t \,\dd t + \psi \gamma_t\, \dd W_t$,
	\item CKLS process (CKLS):   $\dd \gamma_t =( b -\alpha \gamma_t) \,\dd t + \psi \gamma_t^{v}\, \dd W_t$.
\end{itemize}
The BM model is the simplest stochastic model, in which the increment of $\gamma_t$ only has a constant downward trend $b$ without any mean-reverting effect.  
The GBM model is another possibly model for $\gamma_t$. Assuming that the stock price $P_t$ obeys the geometric Brownian motion with drift of $\mu$ and diffusion of $\sigma$ taken as a classical benchmark, this implies that $\gamma_t=E/P_t$ also follows a geometric Brownian motion with  drift of $\alpha=(\mu-\sigma^2)$ and diffusion of $\psi=\sigma$, as reported in Proposition \ref{proposition0}.
The CKLS process is also famous in financial applications, in particular to model short-term interest (see \citep{CKLS}). The CIR process is a special case 
of the CKLS process obtained for $v=0.5$. Hence, choosing the CKLS as an alternative model can help determine whether the non-square root of $\gamma_t$ in 
the diffusion term has more explanatory power.  

\begin{table}[!p]
\begin{tabular}{@{}cccccc@{}}
\toprule
     & NASDAQ 2000                                                                       & S\&P500 1987                                          & Dow Jones 1929                                         & SSEC 2008                                                                 & SSEC 2015                                                     \\ \midrule
\multicolumn{6}{c}{Panel A: Kullback-Leibler Divergence}                                                                                                                                                                                                                                                                                                        \\[1ex]\hline
BM   &\renewcommand\arraystretch{1} \begin{tabular}[c]{@{}c@{}}6.76\\ (0.149)\end{tabular}                            &\renewcommand\arraystretch{1} \begin{tabular}[c]{@{}c@{}}2.79\\ (0.593)\end{tabular} &\renewcommand\arraystretch{1} \begin{tabular}[c]{@{}c@{}}3.80\\ (0.434)\end{tabular} &\renewcommand\arraystretch{1} \begin{tabular}[c]{@{}c@{}}6.83\\ (0.145)\end{tabular}                    &\renewcommand\arraystretch{1} \begin{tabular}[c]{@{}c@{}}15.3\\ (0.004)$^{**}$\end{tabular} \\[2ex]
GBM  &\renewcommand\arraystretch{1} \begin{tabular}[c]{@{}c@{}}8.48\\ (0.076)\end{tabular}                            &\renewcommand\arraystretch{1} \begin{tabular}[c]{@{}c@{}}3.59\\ (0.464)\end{tabular} &\renewcommand\arraystretch{1} \begin{tabular}[c]{@{}c@{}}5.35\\ (0.254)\end{tabular} &\renewcommand\arraystretch{1} \begin{tabular}[c]{@{}c@{}}$1.36\times 10^2$\\ ($<0.001$)$^{***}$\end{tabular} &\renewcommand\arraystretch{1} \begin{tabular}[c]{@{}c@{}}11.0\\ (0.02)$^*$\end{tabular}     \\[2ex]
CKLS &\renewcommand\arraystretch{1} \begin{tabular}[c]{@{}c@{}}0.01\\ (\textgreater{}0.999)\end{tabular}              &\renewcommand\arraystretch{1} \begin{tabular}[c]{@{}c@{}}0.74\\ (0.946)\end{tabular} &\renewcommand\arraystretch{1} \begin{tabular}[c]{@{}c@{}}9.28\\ (0.054)\end{tabular} &\renewcommand\arraystretch{1} \begin{tabular}[c]{@{}c@{}}$1.05\times 10^3$\\ ($<0.001$)$^{***}$\end{tabular} &\renewcommand\arraystretch{1} \begin{tabular}[c]{@{}c@{}}25.1\\ ($<0.001$)$^{***}$\end{tabular}  \\[1ex]\hline
\multicolumn{6}{c}{Panel B: Balakrishnan-Sanghvi Divergence}                                                                                                                                                                                                                                                                                                    \\[1ex]\hline
BM   &\renewcommand\arraystretch{1} \begin{tabular}[c]{@{}c@{}}4.20\\ (0.379)\end{tabular}                            &\renewcommand\arraystretch{1} \begin{tabular}[c]{@{}c@{}}1.50\\ (0.827)\end{tabular} &\renewcommand\arraystretch{1} \begin{tabular}[c]{@{}c@{}}1.99\\ (0.738)\end{tabular} &\renewcommand\arraystretch{1} \begin{tabular}[c]{@{}c@{}}4.36\\ (0.359)\end{tabular}                    &\renewcommand\arraystretch{1} \begin{tabular}[c]{@{}c@{}}4.81\\ (0.308)\end{tabular}        \\[2ex]
GBM  &\renewcommand\arraystretch{1} \begin{tabular}[c]{@{}c@{}}5.99\\ (0.200)\end{tabular}                            &\renewcommand\arraystretch{1} \begin{tabular}[c]{@{}c@{}}1.29\\ (0.863)\end{tabular} &\renewcommand\arraystretch{1} \begin{tabular}[c]{@{}c@{}}2.67\\ (0.615)\end{tabular} &\renewcommand\arraystretch{1} \begin{tabular}[c]{@{}c@{}}5.99\\ (0.200)\end{tabular}                    &\renewcommand\arraystretch{1} \begin{tabular}[c]{@{}c@{}}7.29\\ (0.121)\end{tabular}        \\[2ex]
CKLS &\renewcommand\arraystretch{1} \begin{tabular}[c]{@{}c@{}}$3.30\times 10^{-3}$\\ (\textgreater{}0.999)\end{tabular} &\renewcommand\arraystretch{1} \begin{tabular}[c]{@{}c@{}}0.32\\ (0.988)\end{tabular} &\renewcommand\arraystretch{1} \begin{tabular}[c]{@{}c@{}}4.63\\ (0.360)\end{tabular} &\renewcommand\arraystretch{1} \begin{tabular}[c]{@{}c@{}}9.03\\ (0.060)\end{tabular}                    &\renewcommand\arraystretch{1} \begin{tabular}[c]{@{}c@{}}4.89\\ (0.299)\end{tabular}        \\[1ex]\hline
\multicolumn{6}{c}{Panel C: Rathie-Kannappan Divergence}                                                                                                                                                                                                                                                                                                        \\[1ex]\hline
BM   &\renewcommand\arraystretch{1} \begin{tabular}[c]{@{}c@{}}3.66\\ (0.453)\end{tabular}                            &\renewcommand\arraystretch{1} \begin{tabular}[c]{@{}c@{}}1.42\\ (0.840)\end{tabular} &\renewcommand\arraystretch{1} \begin{tabular}[c]{@{}c@{}}1.92\\ (0.750)\end{tabular} &\renewcommand\arraystretch{1} \begin{tabular}[c]{@{}c@{}}3.77\\ (0.438)\end{tabular}                    &\renewcommand\arraystretch{1} \begin{tabular}[c]{@{}c@{}}6.87\\ (0.143)\end{tabular}        \\[2ex]
GBM  &\renewcommand\arraystretch{1} \begin{tabular}[c]{@{}c@{}}4.95\\ (0.293)\end{tabular}                            &\renewcommand\arraystretch{1} \begin{tabular}[c]{@{}c@{}}1.65\\ (0.799)\end{tabular} &\renewcommand\arraystretch{1} \begin{tabular}[c]{@{}c@{}}2.67\\ [2ex](0.614)\end{tabular} &\renewcommand\arraystretch{1} \begin{tabular}[c]{@{}c@{}}40.3\\ ($<0.001$)$^{***}$\end{tabular}              &\renewcommand\arraystretch{1} \begin{tabular}[c]{@{}c@{}}6.19\\ (0.185)\end{tabular}        \\
CKLS &\renewcommand\arraystretch{1} \begin{tabular}[c]{@{}c@{}}3.36\\ (\textgreater{}0.999)\end{tabular}              &\renewcommand\arraystretch{1} \begin{tabular}[c]{@{}c@{}}0.36\\ (0.986)\end{tabular} &\renewcommand\arraystretch{1} \begin{tabular}[c]{@{}c@{}}4.65\\ (0.325)\end{tabular} &\renewcommand\arraystretch{1} \begin{tabular}[c]{@{}c@{}}$2.3\times 10^2$\\ ($<0.001$)$^{***}$\end{tabular}  &\renewcommand\arraystretch{1} \begin{tabular}[c]{@{}c@{}}10.2\\ (0.372)\end{tabular}        \\[1ex] \bottomrule
\multicolumn{6}{l}{ $p<0.001$: $^{***}$, $p<0.01$: $^{**}$, $p<0.05$: $^{*}$}  
\end{tabular}
\caption{Results of the  $\phi$-divergence tests that $\gamma_t$ follows the CIR process for five historical bubbles. In each bubble case, the time series of $\gamma_t$ is obtained from the stock index price over the time window with $t_{\text{start}}$ and $t_{\text{end}}$ listed in table 1, and the 
corresponding parameters $\pmb{\theta}_0$ of $H_0$ are set to the estimated values $\hat{b}, \hat{\alpha}, \hat{\psi}$. Panels A, B and C respectively give the results 
of the tests for the three $\phi$-divergence measures. Each row lists the values of the $T_{\phi, n}$  statistic and the p-value (below within parenthesis) for 
the $\widehat{\pmb{\theta}}$ estimated by maximization of the quasi-log-likelihood for the corresponding alternative model. } \label{T:2}
\end{table}

We calculate  the statistic $T_{\phi,n}$ for each of three different $\phi(x)$ and then perform the corresponding tests. The three $\phi(x)$ correspond respectively to the KL divergence, BS divergence and RK divergence listed above. 

 As shown in table 2, the tests for NASDAQ 2000, S\&P500 87 and Dow Jones 1929  does not reject the null $H_0$ that $\gamma_t$ obeys the CIR process, regardless of the alternative model to be compared with and for all $\phi$-divergence test statistics. For these three time series, the test results are thus clearly supporting our proposed CIR model for the dynamics of $\gamma_t$  given by eq. (\ref{E: CIRexp}). 
 
However, for the two Chinese stock market bubbles, the three $\phi$-divergence tests give different results. For the case of SSEC 2008, only the test based on 
the BS divergence suggests that the CIR process is more competent than other processes to describe the obtained empirical $\gamma_t$.
In contrast, the KL and RK divergence tests suggest a sharp rejection of $H_0$ in favour of the GBM model or the more general process CKLS 
with an exponent $v$ different from $0.5$.  For SSEC 2015, both BS and RK divergence tests indicate that $H_0$ should not be rejected, 
while the test based on KL divergence suggests the opposite result. Note however in this case that
the null $H_0$ is not strongly rejected with respect to the alternative BM and GBM models. 

Overall, the above tests imply that the proposed model might be better suited to characterize US than Chinese stock markets. However, the comparison 
between the competing models for the three different divergence measures for SSEC 2015 and SSEC 2008 also suggests 
that our model could be applicable to Chinese stock markets.

\section{Conclusion}

We have introduced a new formulation to understand the properties of stock market prices, which is based on 
interpreting the earning-over-price, called here the earning yield process $\gamma_t$, as the key variable driving prices.
$\gamma_t$ has been argued to represent the instantaneous collective belief on the weight of current earnings in pricing the stock. 
In a second interpretation, we have proposed that the earning yield of a company is analogous to the yield-to-maturity of an equivalent perpetual bond 
and $\gamma_t$ stands for the market instantaneous collective belief on such yield-to-maturity.

Our main proposal has been to focus on this $\gamma_t$, in essence the inverse of the price ($\gamma_t \simeq 1/P_t$),
as a more suitable variable to rationalise the many stylised facts of financial prices and returns.
We have illustrated the power of the choice of $\gamma_t \simeq 1/P_t$
as the \emph{right variable} to model financial prices by exploring the properties resulting from perhaps 
the simplest choice for $\gamma_t$ in the form of the CIR process.
This choice was made for simplicity and because the CIR process is the prevailing one to characterise short-term interest properties,
and the earning yield has been argued to be an interest-like variable. 
  
One of the merits of the CIR model for  $\gamma_t$ stems from 
the fact that we have been able to derive explicit analytical solutions for many properties of the resulting price process.
In particular, we have obtained the analytical expression for the 
 transition density for the price process expressed in terms of the modified Bessel function of the first kind. 
 We have obtained the analytical expressions for the first and second moments of prices and cumulative returns, expressed 
 in terms of the Kummer confluent hypergeometric function. We have obtained the ensemble average of the cumulative price.
 We have shown that the inverse CIR process of the price is an ergodic Markovian process.
 We have classified the existence of two regimes: (i) the non-explosive nonlinear regime
 valid for $2\alpha \gamma^* > \psi^2$, in which the price process remains bounded at all times and is also stationary and ergodic;
 (ii) the recurrent explosive bubble regime for $2\alpha \gamma^* \leq \psi^2$, in which 
there is a strictly positive probability for the price to diverge in finite time and then rebound immediately to a finite values.
This implies that the price can transiently explode to infinity. This behavior provides a natural model for explosive financial bubbles.
But even in the non-explosive nonlinear regime, we demonstrated that the price exhibits transient super-exponential behavior
that have been argued to be characteristic signatures of financial bubbles. 
Our model provides a very simple framework, which is quite remarkable in its ability in rationalising 
empirical observations that have not hitherto received explanations, such as the transition of the power law tail exponent  
for the cross-sectional distribution of stock prices
from values larger than $1$ to a value converging to $1$ close to the end of the Japanese Internet bubble.
  
We have developed a quasi-maximum likelihood method with the L-BFGS-B optimisation algorithm to calibrate the model
to five well-known historical bubbles in the US and China stock markets. The estimated $\alpha$ and $\psi$
parameters  for the five bubbles are within the typical interval that accords with the numerical analysis. Furthermore, for the five bubbles, the perceived intrinsic value $P^*$ has been estimated to be close to the peak price in the calibration window, except for the second Chinese bubble. 
This implies that estimating $P^*$ can help determine the degree of maturation of a bubble and should aid in constructing early-warning systems for bubble collapse. 
 
 Finally, let us reflect briefly in more general terms on the profound implications of our proposition that the inverse of the price ($\gamma_t \simeq 1/P_t$)
is the ``right'' financial variable to model. In the systematic scientific investigation of the world, 
a key component is identifying the right variable to model, which involves selecting the key factors 
that are relevant to the phenomenon being studied and developing a model that accurately captures their interactions and relationships.
Identifying the right variable to model is a critical step in the scientific process and 
can have a significant impact on the accuracy and usefulness of the resulting models.  
This translates into the quest of finding the ``right'' variable that makes the theory simple and insightful.
In the present case, we have proposed that the ``right'' variable is not the price but its inverse. In fact, 
from an epistemological perspective, we proposed that, in the causal relationship between price and earning yield, the earning yield is the cause and the price, as expressed through the functional relationship $P_t \simeq 1/\gamma_t$, is the effect. The dynamics of price is determined by the variations in earning yield, not vice versa.
In the physical sciences, it has been documented many times
that observable variables are not necessarily the ``right'' variables in terms of which the theory is simple.
Simultaneously, adhering to the principle of identifying the correct ``cause" and modeling it specifically provides the modeler with guidance on how to construct better and more parsimonious reduced models.
We think that this strategy introduced in the present work can open up a large new set of insights 
in financial economics.

\appendix
\section{Proofs}

\subsection{Proof of Proposition \ref{proposition1} \label{proof-E:Pb}} 

\begin{proof}
	The proof is constructed by directly showing that the original process defined by
expressions \eqref{pricing1} and \eqref{E:bubbleCIRgamma} does satisfy SDE (\ref{E:Pb}). 
For ease of notations, we remove subscript $t$ is neglected. According to the It\^{o} lemma, (\ref{pricing1}) leads to
\begin{align*}
	\dd P &=\dfrac{\pD P}{\pD \gamma}\dd \gamma +\dfrac{1}{2}\dfrac{\pD^2 P}{\pD \gamma^2}(\dd \gamma)^2\\
&= -\dfrac{E}{\gamma^2}~\dd \gamma +\dfrac{1}{2}\cdot 2\cdot \dfrac{E}{\gamma^3}~(\dd \gamma)^2 \\
&=-\dfrac{E}{\gamma}\cdot\dfrac{1}{\gamma}[~-\alpha(\gamma-\gamma^*)~\dd t +\psi \sqrt{\gamma}~\dd W~]+\dfrac{E}{\gamma}\cdot\dfrac{1}{\gamma^2}\cdot \psi^2\gamma~\dd t\\
&= P\left[\alpha\left(1-\dfrac{\gamma^*}{\gamma}\right)+\dfrac{\psi^2}{\gamma}\right]~\dd t - P \psi\sqrt{\dfrac{1}{\gamma}}~\dd W
\end{align*}

Therefore, 
\begin{align*}
\dfrac{\dd P}{P}&=\left[\alpha\left(1-\dfrac{\gamma^*}{\gamma}\right)+\dfrac{\psi^2}{\gamma}\right]~\dd t - \psi\sqrt{\dfrac{1}{\gamma}}~\dd W\\
&\phantom{=} \text{replace}~\gamma ~\text{by}~\dfrac{E}{P} ~\text{and then replace}~\gamma^* ~\text{by}~ \dfrac{E}{P^*}\\
\dfrac{\dd P}{P}&=\left[~\alpha\left(\dfrac{P^*-P}{P^*}\right)+\psi^2\dfrac{P}{E}~\right]\dd t -\psi \sqrt{\dfrac{P}{E}}~\dd W~\\
&\phantom{=} \text {according to the symmetry of }\dd W_t\\
\dfrac{\dd P}{P}&=\left[~\alpha\left(\dfrac{P^*-P}{P^*}\right)+\psi^2\dfrac{P}{E}~\right]\dd t +\psi \sqrt{\dfrac{P}{E}}~\dd W~.
\end{align*}
\end{proof}
{\hfill $\square$}

\subsection{Proof of Proposition \ref{proposition2} \label{proof-martin-SDF}}  

\begin{proof}
The proposition is proven by checking that $m_t$ satisfies the two conditions for qualifying as a stochastic discount factor (SDF). 

Firstly, it is easy to see that $\E_t(\dd m_t)=0$. As $Z^{(i)}_t$ is a Gauss-Winner process,  $\E_t(m_t\vartheta_t^{(i)}\dd Z^{(i)}_t)=m_t\vartheta_t^{(i)}\cdot\E_t(\dd Z^{(i)}_t)=0$. Hence, $\E_t(\dd m_t)=\sum_{i=1}^N \E_t(m_t\vartheta_t^{(i)}\dd Z^{(i)}_t)=0$.

We then check whether the prices of an arbitrary stock and its derivative satisfy the second condition. 

For stock $i$, we have 
\begin{align*}
	\E_t\left(\dfrac{\dd m_t}{m_t}\dfrac{\dd P^{(i)}_t}{P^{(i)}_t}\right)&=\E_t\left(\vartheta^{(i)}_t\dd Z^{(i)}_t\psi^{(i)}\sqrt{\dfrac{P^{(i)}_t}{E^{(i)}}}\dd W^{(i)}_t+\sum_{j\neq i }\vartheta^{(j)}_t\dd Z^{(j)}_t\psi^{(i)}\sqrt{\dfrac{P^{(i)}_t}{E^{(i)}}}\dd W^{(i)}_t\right)\\
&=\vartheta^{(i)}_t\psi^{(i)}\sqrt{\dfrac{P^{(i)}_t}{E^{(i)}}}\rho^{(i)}\delta_{ii}\dd t+\sum_{j\neq i }\vartheta^{(j)}_t\psi^{(i)}\sqrt{\dfrac{P^{(i)}_t}{E^{(i)}}}\rho^{(i)}\delta_{ij}\dd t\\
&=\vartheta^{(i)}_t\rho^{(i)}\cdot \psi^{(i)}\sqrt{\dfrac{P^{(i)}_t}{E^{(i)}}}\,\dd t, \\
\E_t\left(\dfrac{\dd P^{(i)}_t}{P^{(i)}_t}\right)&=\left[\alpha^{(i)} +\left(\dfrac{\psi^{(i)2}}{E^{(i)}}-\dfrac{\alpha^{(i)}}{P^{*(i)}}\right)\right]\dd t.
\end{align*}
If the following equality holds
\begin{equation*}
	\vartheta^{(i)} _t \rho^{(i)} =\frac{[r_f-\alpha^{(i)}]-\psi^{(i)2}}{\psi^{(i)}}\frac{1}{\sqrt{\smash[b] {P_t^{(i)}E^{(i)}}}}+\frac{\alpha^{(i)}\sqrt{\smash[b]{P_t^{(i)}E^{(i)}}}}{P^{*(i)}\psi^{(i)}}~,
\end{equation*}
then the second condition in Equation (\ref{E: martingale2}) for all stocks is also satisfied.

We then check the conditions for the prices of derivatives. For ease of notation, we denote 
\begin{align*}
	 \mu^{(i)}_t:&=\mu(P_t^{(i)})=\left[~\alpha^{(i)}\left(\dfrac{P^{(i)*}-P^{(i)}_t}{P^{(i)*}}\right)+\psi^{(i) 2}~\dfrac{P^{(i)}_t}{E^{(i)}}~\right],\\
     \sigma^{(i)}_t:&=\sigma( P_t^{(i)})=\psi^{(i)} \sqrt{\dfrac{P^{(i)}_t}{E^{(i)}}}. 
\end{align*}
According to the It\^{o} formula, we have,
\begin{equation}
	\dd\xi^{(i)}_t= \left[\dfrac{\pD \xi^{(i)}_t}{\pD t}+\dfrac{\pD \xi^{(i)}_t}{\pD P^{(i)}_t} \mu^{(i)}_t P^{(i)}_t+\dfrac{1}{2}\dfrac{\pD^2\xi^{(i)}_t}{\pD P^{(i)2}_t} \sigma^{(i)2}_t P^{(i)2}_t\right]\dd t +\dfrac{\pD \xi^{(i)}_t}{\pD P^{(i)}_t} P^{(i)}_t \sigma^{(i)}_t \dd W^{(i)}_t.
\end{equation}
At any time $t$, if one take $1$ share of derivative $\xi^{(i)}_t$ and simultaneously short $\dfrac{\pD \xi^{(i)}_t}{\pD P^{(i)}_t}$ shares at the price $P^{(i)}_t$, the diffusion term can be neutralised, and the portfolio must be risk-free. Further, given the no-arbitrage condition that any risk-free portfolio in the market must grow 
at the risk-free rate $r_f$, we get 
\begin{equation*}
	\left[\dfrac{\pD \xi^{(i)}_t}{\pD t}+\dfrac{\pD \xi^{(i)}_t}{\pD P^{(i)}_t} \mu^{(i)}_t P^{(i)}_t+\dfrac{1}{2}\dfrac{\pD^2\xi^{(i)}_t}{\pD P^{(i)2}_t} \sigma^{(i)2}_t P^{(i)2}_t\right]-\dfrac{\pD \xi^{(i)}_t}{\pD P} \mu^{(i)}_t P^{(i)}_t=r_f \left(\xi^{(i)}_t-\dfrac{\pD \xi^{(i)}_t}{\pD P^{(i)}_t} P^{(i)}_t \right).
\end{equation*}
Moreover,
\begin{equation}\label{E: optioneq}
	\dfrac{\pD \xi^{(i)}_t}{\pD t}+\dfrac{\pD \xi^{(i)}_t}{\pD P^{(i)}_t} r_f P^{(i)}_t+\dfrac{1}{2}\dfrac{\pD^2\xi^{(i)}_t}{\pD P^{(i)2}_t} \sigma^{(i)2}_t P^{(i)2}_t=r_f \xi^{(i)}_t.
\end{equation}
This equation can be considered as the principal motion equation for the price of derivatives. Thus,
\begin{align*}
	\E_t(\dd\xi^{(i)}_t)&= \left[\dfrac{\pD \xi^{(i)}_t}{\pD t}+\dfrac{\pD \xi^{(i)}_t}{\pD P^{(i)}_t} \mu^{(i)}_t P^{(i)}_t+\dfrac{1}{2}\dfrac{\pD^2\xi^{(i)}_t}{\pD P^{(i)2}_t} \sigma^{(i)2}_t P^{(i)2}_t\right]\dd t\\
\E_t\left(\dfrac{\dd m_t}{m_t}\dd \xi^{(i)}_t\right)&=\E_t\left(\vartheta^{(i)}_t\dd Z^{(i)}_t\dfrac{\pD \xi^{(i)}_t}{\pD P^{(i)}_t} P^{(i)}_t \sigma^{(i)}_t \dd W^{(i)}_t+\sum_{j\neq i }\vartheta^{(j)}_t\dd Z^{(j)}_t\dfrac{\pD \xi^{(i)}_t}{\pD P^{(i)}_t} P^{(i)}_t \sigma^{(i)}_t \dd W^{(i)}_t\right)\\
&=\vartheta^{(i)}_t   \dfrac{\pD \xi^{(i)}_t}{\pD P^{(i)}_t} P^{(i)}_t \sigma^{(i)}_t    \rho^{(i)}\delta_{ii}\dd t+\sum_{j\neq i }\vartheta^{(j)}_t  \dfrac{\pD \xi^{(i)}_t}{\pD P^{(i)}_t} P^{(i)}_t \sigma^{(i)}_t   \rho^{(i)}\delta_{ij}\dd t\\
&=\vartheta^{(i)}_t     \rho^{(i)}\cdot  \dfrac{\pD \xi^{(i)}_t}{\pD P^{(i)}_t} P^{(i)}_t \sigma^{(i)}_t \dd t.
\end{align*}
It is important to note that $\vartheta^{(i)}_t     \rho^{(i)}=\dfrac{r_f-\mu^{(i)}_t}{\sigma^{(i)}_t}$. Hence, 
\begin{align*}
	\E_t(\dd\xi^{(i)}_t)+\E_t\left(\dfrac{\dd m_t}{m_t}\dd \xi^{(i)}_t\right)&=\left[\dfrac{\pD \xi^{(i)}_t}{\pD t}+\dfrac{\pD \xi^{(i)}_t}{\pD P} r_f P^{(i)}_t+\dfrac{1}{2}\dfrac{\pD^2\xi^{(i)}_t}{\pD P^{(i)2}_t} \sigma^{(i)2}_t P^{(i)2}_t\right]\dd t.
\end{align*}
Recall that, from Equation (\ref{E: optioneq}). the right-hand-side of the above equation is equal to $r_f \xi^{(i)}_t \dd t$. Therefore, 
\begin{align*}
	\E_t\left(\dfrac{\dd\xi^{(i)}_t}{\xi^{(i)}_t}\right)+\E_t\left(\dfrac{\dd m_t}{m_t}\dfrac{\dd\xi^{(i)}_t}{\xi^{(i)}_t}\right)=r_f\dd t~.
\end{align*}
\end{proof}
{\hfill $\square$}

\subsection{Proof of expression eq.(\ref{E:csdecompforlng}) with Campbell-Shiller decomposition}\label{A:laow0nn}

The term $\ln\left(1+\dfrac{1}{c_{t+1}\gamma_{t+1}}\right)$ in equation (\ref{E:lngamma}) can be rewritten as $\ln(1+\exp(-\ln c_{t+1}\gamma_{t+1}))$. 
Performing a first-order Taylor expansion of the function $f(x)=\ln (1+\exp(x))$ around $\overline{\ln c\gamma}:=\E(\ln c_t \gamma_t)$, one has
\begin{align}\label{E:taylorexpan}
\ln\left(1+\dfrac{1}{c_{t+1}\gamma_{t+1}}\right) & \approx \ln(1+e^{-\overline{\ln c\gamma}}) - \dfrac{e^{-\overline{\ln c\gamma}}}{1+e^{-\overline{\ln c\gamma}}}	(\ln \gamma_{t+1}+\ln c_{t+1}-\overline{\ln c\gamma})\notag\\
&=\kappa - \rho (\ln \gamma_{t+1}+\ln c_{t+1})
\end{align}
where, 
\begin{equation}
	\begin{cases}
		\kappa = \ln(1+e^{-\overline{\ln c\gamma}}) + \dfrac{e^{-\overline{\ln c\gamma}}}{1+e^{-\overline{\ln c\gamma}}}	\cdot \overline{\ln c\gamma}\\
		\rho = \dfrac{e^{-\overline{\ln c\gamma}}}{1+e^{-\overline{\ln c\gamma}}}	~. 
	\end{cases}
\end{equation}

By plugging eq. (\ref{E:taylorexpan}) into eq. (\ref{E:lngamma}), we obtain
\begin{equation}\label{E:iterlngamma}
	\ln\gamma_t=-\kappa + r_{t+1}+\rho \ln \gamma_{t+1} -\Delta \ln E_{t+1}-(1-\rho) \ln c_{t+1} 
\end{equation}
where $r_t$ represents $\ln R_t$ and $-\Delta \ln E_{t+1}$ represents $\ln E_{t+1}-\ln E_{t}$.  Iterating eq.(\ref{E:iterlngamma}) forward $n$ steps yields
\begin{multline}
	\ln \gamma_t = -\kappa \sum_{i=0}\rho^i +( \sum_{i=0}^{n}\rho^i  \mathscr{L}^i)r_{t+1}-( \sum_{i=0}^{n}\rho^i  \mathscr{L}^i)\Delta\ln E_{t+1}\\-(1-\rho)(\sum_{i=0}^{n}\rho^i  \mathscr{L}^i)\ln c_{t+1}+\rho^{n+1}\ln \gamma_{t+n+1} ,
\end{multline}
where $\mathscr{L}$ denotes the lead operator that transforms $x_{t}$ into $x_{t+1}$. 

As $n$ approaches infinity and considering the transversality condition $\lim_{n\to \infty}\rho^n \ln \gamma_{t+n}=0$ as well as  the logit term $\rho\in (0,1)$, we finally obtain
\begin{equation}
	\ln \gamma_t =-\dfrac{\kappa}{1-\rho}+\dfrac{r_{t+1}}{1-\rho\mathscr{L}}-\dfrac{\Delta\ln E_{t+1}+(1-\rho)\ln c_{t+1}}{1-\rho\mathscr{L}}~.
\end{equation}

\subsection{Proof of Proposition \ref{proptranden}} \label{proof-transition-prob}

Starting from expression (\ref{defgamma}) relating $\gamma_t$ to $P_t$, 
the transition probability $f(P_t,t\mid P_0)$ for $P_t$ is obtained from the transition probability $f^{CIR}(\gamma_t, t\mid \gamma_0)$ for $\gamma_t$ 
via the standard condition of the conservation of probability under a change of variable, which yields
\begin{equation}
f(P_t,t\mid P_0)=f^{CIR}(\gamma_t, t\mid \gamma_0)\cdot \dfrac{E}{P_t^2}
=f^{CIR}\left(\dfrac{E}{P_t}, t\mid \dfrac{E}{P_0}\right)\cdot \dfrac{E}{P_t^2}.
\label{konkav}
\end{equation}
The well-known CIR conditional transition density is given by
\begin{equation}\label{E:CIRdensity0}
f^{CIR}(\gamma_t, t\mid \gamma_0)=d~ e^{-(u+v)}\left(\dfrac{v}{u}\right)^{\frac{q}{2}}I_q(\sqrt{2uv}),	
\end{equation}
where 
\begin{align*}
d&=\dfrac{2\alpha}{\psi^2(1-e^{-\alpha t})}, \qquad q=\dfrac{2\alpha \gamma^*}{\psi^2}-1\\
u&= d\,\gamma_0\,e^{-\alpha t},  \qquad \qquad v=d\, \gamma_t	~.
\end{align*}
Substituting $\gamma^*=\dfrac{E}{P^*}$, $\gamma_0=\dfrac{E}{P_0}$, and $\gamma_t=\dfrac{E}{P_t}$ into $q, u, v$, respectively, we obtain
\begin{align*}
q&=\dfrac{2\alpha E}{\psi^2 P^*}-1=\dfrac{H}{P^*}-1 \\
u&=\dfrac{2\alpha E~ e^{-\alpha t}}{\psi^2 P_0 ~(1-e^{-\alpha t})}=\dfrac{H}{P_0}\cdot \dfrac{1}{e^{\alpha t}-1}\\
v&=\dfrac{2\alpha E }{\psi^2 (1-e^{-\alpha t}) P_t}=\dfrac{H}{P_t}\cdot\dfrac{1}{1-e^{-\alpha t}}.
\end{align*}
From Equations (\ref{konkav}) and (\ref{E:CIRdensity0}), we have
\begin{align*}
f(P_t, t\mid P_0)&=\dfrac{2\alpha}{\psi^2(1-e^{-\alpha t})}~ e^{-(u+v)}\left(\dfrac{v}{u}\right)^{\frac{q}{2}}I_q(\sqrt{2uv})\cdot\dfrac{E}{P_t^2}\\
&=\dfrac{2\alpha E}{\psi^2(1-e^{-\alpha t})}~ e^{-(u+v)} v^{\frac{q}{2}}\cdot\dfrac{1}{P_t^2}\cdot u^{-\frac{q}{2}}~I_q(\sqrt{2uv})\\
&=\dfrac{H}{1-e^{-\alpha t}}~ e^{-(u+v)} v^{\frac{q}{2}+2}\cdot\left(\dfrac{H}{1-e^{-\alpha t}}\right)^{-2}\cdot u^{-\frac{q}{2}}~I_q(\sqrt{2uv})\\
&=c~e^{-(u+v)}v^{\frac{q}{2}+2}u^{-\frac{q}{2}}~I_q(2\sqrt{uv}).
\end{align*}

\subsection{Proof of the ergodicity of the price process (\ref{E:Pb})}\label{A:ergodic}

The analysis follows \cite{iacus2009} and \cite{Cherny2005ia}, examining the integral convergence conditions of the so-called \emph{scale measure} and \emph{speed measure}. Specifically, consider a Markovian process in the form of a stochastic differential equation (SDE)
\begin{equation*}
	\dd X_t=b(X_t)~\dd t+\sigma(X_t)~\dd W~,
\end{equation*}
the scale measure and speed measure are defined respectively as
\begin{align*}
	\rho(x)&=e^{-\int_a^x \frac{2 b(y)}{\sigma^2(y)} \dd y}~,\qquad \text{arbitrary}~~a\ge 0~~\text{and}~~a<x\\
	m(x)&=\dfrac{1}{\sigma^2(x)\rho(x)}~.
\end{align*}
If $\int_0^{\infty} m(x)~\dd x <\infty$, the above process is ergodic inside the domain $[\,0,\infty\,]$. Moreover, if 
\begin{align*}
	&\int_0^a \rho(x)\dd x =\infty, \quad \int_a^{\infty} \rho(x)\dd x =\infty,\qquad \text{arbitrary}~~a\ge 0~~\text{and}~~a<\infty,
\end{align*}
the process can never reach $0$ or $\infty$ with a positive probability. However, If one or both integrals are convergent, the corresponding boundary $0$ or $\infty$ is reached and is instantaneously reflecting.

For the price process specified by Equation (\ref{E:Pb}), the scale measure (taking $a=1$) is
\begin{equation*}
	\rho(x)=\exp\left\{-\int_1^x\dfrac{2(\alpha+A y)y}{B y^3}\dd y\right\}=e^{-\frac{2\alpha}{B}}\cdot e^{\frac{2\alpha}{B}\frac{1}{x}}x^{-\frac{2A}{B}},
\end{equation*}
where $A=\dfrac{\psi^2}{E}-\dfrac{\alpha}{P^*}$ and $B=\dfrac{\psi^2}{E}$ with $A\in (-\dfrac{\alpha}{P^*},\,B), ~B \in (0,\infty)$. Thus,
\begin{equation*}
	m(x)=\dfrac{e^{\frac{2\alpha}{B}}}{B}\cdot e^{-\frac{2\alpha}{B}\frac{1}{x}}x^{\frac{2A}{B}-3}~.
\end{equation*}
Then,
\begin{align*}
	\int_0^{\infty} m(x)~\dd x & =\dfrac{e^{\frac{2\alpha}{B}}}{B}\int_0^{\infty}e^{-\frac{2\alpha}{B}\frac{1}{x}}x^{\frac{2A}{B}-3}~\dd x\\
&=\dfrac{e^{\frac{2\alpha}{B}}}{B}\int_0^{\infty} e^{-\frac{2\alpha}{B}\cdot y}y^{-\frac{2A}{B}+1}~\dd y\qquad \text{with $y=\dfrac{1}{x}$}\\
&=\dfrac{e^{\frac{2\alpha}{B}}}{B}\int_0^{\infty} e^{-z} \left(\dfrac{B}{2\alpha}\right)^{1-\frac{2A}{B}}\cdot z^{-\frac{2A}{B}+1}\cdot \dfrac{B}{2\alpha}~\dd z\qquad \text{with $z=\dfrac{2\alpha}{B}$}y\\
&=\dfrac{e^{\frac{2\alpha}{B}}}{B}\left(\dfrac{B}{2\alpha}\right)^{2-\frac{2A}{B}}\Gamma(2-\frac{2A}{B})<\infty~.
\end{align*}
Therefore, the price process is ergodic in the domain $[\,0,+\infty\,]$. 
Since that the price process is the reciprocal of the CIR process, this result is intimately associated with the ergodicity for the CIR process.
Additionally, 
\begin{align*}
\int_1^{\infty} \rho(x)\dd x &=e^{-\frac{2\alpha}{B}} \int_1^{\infty}  e^{\frac{2\alpha}{B}\frac{1}{x}}x^{-\frac{2A}{B}}~\dd x \\
&=e^{-\frac{2\alpha}{B}}\cdot \left(\dfrac{2\alpha}{B}\right)^{1-\frac{2A}{B}}\int_{0}^{\frac{2\alpha}{B}} e^{y}~y^{-2+\frac{2A}{B}}~\dd y \qquad \text{with $y=\dfrac{2\alpha}{B x}$}\\
&=e^{-\frac{2\alpha}{B}}\cdot \left(\dfrac{2\alpha}{B}\right)^{1-\frac{2A}{B}}\cdot(-1)^{-2+\frac{2A}{B}} \int^{0}_{-\frac{2\alpha}{B}} e^{-z}z^{-2+\frac{2A}{B}}~\dd z\qquad \text{with $z=-y$}\\
&=(-1)^{-2+\frac{2A}{B}}e^{-\frac{2\alpha}{B}}\cdot \left(\dfrac{2\alpha}{B}\right)^{1-\frac{2A}{B}}\left(\int^{\infty}_{-\frac{2\alpha}{B}} e^{-z}z^{-2+\frac{2A}{B}}~\dd z-\int^{\infty}_{0} e^{-z}z^{-2+\frac{2A}{B}}~\dd z\right)\\
&=\begin{cases}
C_1\left[\,\Gamma(-1+\frac{2A}{B}, -\frac{2\alpha}{B})-\Gamma(-1+\frac{2A}{B})\,\right]<\infty	,\qquad -1+\frac{2A}{B}>0,\\
\infty
\end{cases}
\end{align*} 
where $C_1:=(-1)^{-2+\frac{2A}{B}}e^{-\frac{2\alpha}{B}}\cdot \left(\frac{2\alpha}{B}\right)^{1-\frac{2A}{B}}$. The integral $\int_1^{\infty} \rho(x)\dd x$ converges when the condition $-1+\frac{2A}{B}>0$ is satisfied. This condition can be finally translated into $\frac{2\alpha E}{\psi^2}< P^*$, i.e.,~$H<P^*$.
However, 
\begin{align*}
\int_0^{1} \rho(x)\dd x &=e^{-\frac{2\alpha}{B}} \int_0^{1}  e^{\frac{2\alpha}{B}\frac{1}{x}}x^{-\frac{2A}{B}}~\dd x \\
&=e^{-\frac{2\alpha}{B}}\cdot \left(\dfrac{2\alpha}{B}\right)^{1-\frac{2A}{B}}\int_{1}^{\infty} e^{y}~y^{-2(1-\frac{A}{B})}~\dd y \qquad \text{with $y=\dfrac{2\alpha}{B x}$}\\
&=\infty~.
\end{align*}
Hence, the price process is ergodic and can never reach $zero$. When $H<P^*$, the process has a strictly positive probability of exploding to infinity 
in finite time. When this happens, the price is instantaneously reflected from infinity. Otherwise ($H>P^*$), the price does not reach infinity in finite time.

\subsection{Proof of Proposition \ref{prop:stableden} \label{proof-stablede}}

\begin{proof}
As the price process that is solution of Equation \eqref{E:Pb} is ergodic, it has a unique stationary distribution, which is the limiting distribution 
of its long-term evolution. Using the definitions (\ref{trhtgbq1}-\ref{trhtgbq4}), irrespective of whether $H \geq P^*$ or $H< P^*$  and for an arbitrary $P_0$, we always have   
\begin{align*}
\pi(P)&=\lim_{t\to\infty} f(P_t,t\mid P_0)\\
&=\lim_{t\to\infty} c~e^{-(u+v)}v^{\frac{q}{2}+2}u^{-\frac{q}{2}}~I_q(2\sqrt{uv})\\
&=\lim_{t\to\infty} c~e^{-(u+v)}v^{\frac{q}{2}+2}u^{-\frac{q}{2}}~\sum_{n=0}^{\infty}\left(\dfrac{2\sqrt{uv}}{2}\right)^{2n+q}\dfrac{1}{n!~\Gamma(n+q+1)}\\
&=\lim_{t\to\infty} c~e^{-(u+v)}v^{q+2}~\sum_{n=0}^{\infty}\dfrac{(uv)^n}{n!~\Gamma(n+q+1)}\\
\text{Notice that}& \text{~~$c=\dfrac{1-e^{-\alpha t}}{H}\overset{t\to\infty}{\longrightarrow}\dfrac{1}{H}$, $v=\dfrac{1}{c P}\overset{t\to\infty}{\longrightarrow}\dfrac{H}{P}$, and $u\sim \dfrac{1}{e^{\alpha t}-1}\overset{t\to\infty}{\longrightarrow} 0$, then}\\
\pi(P)&=\dfrac{1}{H}e^{-\frac{H}{P}}\left(\dfrac{H}{P}\right)^{q+2}\dfrac{1}{\Gamma(q+1)}\\ 
&\text{Replacing $q$ by $\frac{H}{P^*}-1$, this yields}\\
&=\dfrac{H^{\frac{H}{P^*}}}{\Gamma\left(\frac{H}{P^*}\right)}e^{-\frac{H}{P}}P^{-(\frac{H}{P^*}+1)}~.
\end{align*} {\hfill $\square$}

\end{proof}

\subsection{Proof of Proposition \ref{rhj243bq} \label{jiu5mori9}}

Using the definitions (\ref{trhtgbq1}-\ref{trhtgbq4}),
\begin{align*}
\E(P_t\mid P_0)&=\int_0^{\infty}P_t\cdot f(P_t, t\mid P_0)~\dd P_t\\
&=\int_0^{\infty}P_t\cdot c \cdot e^{-(u+v)}~v^{\frac{q}{2}+2}~u^{-\frac{q}{2}}~I_q(2\sqrt{uv})~\dd P_t\\
&\text{Note that  $v=\dfrac{1}{c P_t}$ and performing the change of variable $\dd P_t=-c P_t^2~\dd v$ }\\
&=-\int_{\infty}^0(cP_t)^2 P_t\cdot e^{-(u+v)}~v^{\frac{q}{2}+2}~u^{-\frac{q}{2}}~I_q(2\sqrt{uv})~\dd v\\
&=\int_0^{\infty} (cP_t)^2 P_t \cdot e^{-(u+v)}~v^{\frac{q}{2}-1}\cdot\left(\dfrac{1}{cP_t}\right)^3\cdot u^{-\frac{q}{2}}~I_q(2\sqrt{uv})~\dd v\\
&=\int_0^{\infty} \dfrac{1}{c} \cdot e^{-(u+v)}~v^{\frac{q}{2}-1}~ u^{-\frac{q}{2}}~I_q(2\sqrt{uv})~\dd v\\
&=\int_0^{\infty} \dfrac{1}{c} \cdot e^{-(u+v)}~v^{\frac{q}{2}-1}~ u^{-\frac{q}{2}}~\sum_{n=0}^{\infty}\left(\dfrac{2\sqrt{uv}}{2}\right)^{2n+q}\dfrac{1}{n!~\Gamma(n+q+1)}~\dd v\\
&=\int_0^{\infty} \dfrac{1}{c} \cdot e^{-(u+v)}~v^{\frac{q}{2}-1}~ u^{-\frac{q}{2}}~(uv)^{\frac{q}{2}}\sum_{n=0}^{\infty}\dfrac{(uv)^n}{\Gamma(n+1)\Gamma(n+q+1)}~\dd v\\
&=\dfrac{e^{-u}}{c}\sum_{n=0}^{\infty}\left(\dfrac{u^n}{\Gamma(n+1)\Gamma(n+1+q)}\cdot\int_0^{\infty}e^{-v}v^{q-1+n}~\dd v\right)\\
&=\dfrac{e^{-u}}{c}\sum_{n=0}^{\infty}\dfrac{\Gamma(n+q)}{\Gamma(n+1)\Gamma(n+1+q)}\cdot u^n\\
&=\dfrac{e^{-u}}{c}\cdot\dfrac{\Gamma(q)}{\Gamma(q+1)}\cdot{}_1F_1(q,q+1,u)\\
&=\dfrac{e^{-u}}{c~q}\cdot{}_1F_1(q,q+1,u)
\end{align*}
Therefore, $\E(R_t\mid P_0)=\E\left(\dfrac{P_t}{P_0}\mid P_0\right)=\dfrac{e^{-u}}{cqP_0}\cdot{}_1F_1(q,q+1,u)$.$~~~~~~~~~~~~~~~~~~~~~`~~~~~~~~~~~~~~~~~~~~~~~~~\Box$

\subsection{Proof of Proposition \ref{wrhbtbqbq}} \label{A:condvar}   

Based on the explicit conditional transition density of Equation \eqref{E:traden} and definitions (\ref{trhtgbq1}-\ref{trhtgbq4}), we have
\begin{align*}
\E(P^2_t\mid P_0)&=\int_0^{\infty}P^2_t\cdot f(P_t, t\mid P_0)~\dd P_t\\
&=\int_0^{\infty}P^2_t\cdot c \cdot e^{-(u+v)}~v^{\frac{q}{2}+2}~u^{-\frac{q}{2}}~I_q(2\sqrt{uv})~\dd P_t\\
&=-\int_{\infty}^0(cP_t)^2 P^2_t\cdot e^{-(u+v)}~v^{\frac{q}{2}+2}~u^{-\frac{q}{2}}~I_q(2\sqrt{uv})~\dd v\\
&=\int_0^{\infty} (cP_t)^2 P^2_t \cdot e^{-(u+v)}~v^{\frac{q}{2}-2}\cdot\left(\dfrac{1}{cP_t}\right)^4\cdot u^{-\frac{q}{2}}~I_q(2\sqrt{uv})~\dd v\\
&=\int_0^{\infty} \dfrac{1}{c^2} \cdot e^{-(u+v)}~v^{\frac{q}{2}-2}~ u^{-\frac{q}{2}}~I_q(2\sqrt{uv})~\dd v\\
&=\int_0^{\infty} \dfrac{1}{c^2} \cdot e^{-(u+v)}~v^{\frac{q}{2}-2}~ u^{-\frac{q}{2}}~\sum_{n=0}^{\infty}\left(\dfrac{2\sqrt{uv}}{2}\right)^{2n+q}\dfrac{1}{n!~\Gamma(n+q+1)}~\dd v\\
&=\int_0^{\infty} \dfrac{1}{c^2} \cdot e^{-(u+v)}~v^{\frac{q}{2}-2}~ u^{-\frac{q}{2}}~(uv)^{\frac{q}{2}}\sum_{n=0}^{\infty}\dfrac{(uv)^n}{\Gamma(n+1)\Gamma(n+q+1)}~\dd v\\
&=\dfrac{e^{-u}}{c^2}\sum_{n=0}^{\infty}\left(\dfrac{u^n}{\Gamma(n+1)\Gamma(n+1+q)}\cdot\int_0^{\infty}e^{-v}v^{q-2+n}~\dd v\right)\\
&=\dfrac{e^{-u}}{c^2}\sum_{n=0}^{\infty}\dfrac{\Gamma(n+q-1)}{\Gamma(n+1)\Gamma(n+1+q)}\cdot u^n\\
&=\dfrac{e^{-u}}{c^2}\cdot\dfrac{\Gamma(q-1)}{\Gamma(q+1)}\cdot{}_1F_1(q-1,q+1,u)\\
&=\dfrac{e^{-u}}{c^2~q(q-1)}\cdot{}_1F_1(q,q+1,u)
\end{align*}	

\subsection{Proof of Proposition \ref{wrhwth2tybtbqbq}} \label{jiu52th1mori9}
 
\begin{align}
\E (R_{\infty}\mid P_0)&=\dfrac{\E(P_{\infty})}{P_0}=\dfrac{1}{P_0}\int_0^{\infty} P\cdot\pi(P)~\dd P\notag\\
&= \dfrac{1}{P_0}\int_0^{\infty} P \dfrac{H^{\frac{H}{P^*}}}{\Gamma\left(\frac{H}{P^*}\right)}e^{-\frac{H}{P}}P^{-(\frac{H}{P^*}+1)} ~\dd P\notag\\
&=\dfrac{1}{P_0\cdot\Gamma\left(\frac{H}{P^*}\right)}\int_0^{\infty} e^{-\frac{H}{P}}\left(\dfrac{H}{P}\right)^{\frac{H}{P^*}}~\dd P\notag\\
&\text{Defining $z=\dfrac{H}{P}$ and applying change of variable $\dd P=-\dfrac{P^2}{H}\dd z$}\notag\\
&=\dfrac{H}{P_0\cdot\Gamma\left(\frac{H}{P^*}\right)}\int_0^{\infty} e^{-z}~z^{\frac{H}{P^*}-2}~\dd z\notag\\
\E (R_{\infty}\mid P_0)&=\dfrac{H\cdot\Gamma\left(\frac{H}{P^*}-1\right)}{P_0\cdot\Gamma\left(\frac{H}{P^*}\right)}=\dfrac{P^*}{P_0}\cdot\left(1-\dfrac{P^*}{H}\right)^{-1}~.\label{E:condqtelongR}
\end{align}

A similar derivation holds for the variance of the cumulative return.
\begin{align}
 \E (R^2_{\infty}\mid P_0)&=\dfrac{1}{P^2_0}\int_0^{\infty} P^2\cdot\pi(P)~\dd P\notag\\
&=\dfrac{1}{P_0^2\cdot\Gamma\left(\frac{H}{P^*}\right)}\int_0^{\infty} e^{-\frac{H}{P}}\left(\dfrac{H}{P}\right)^{\frac{H}{P^*}}P~\dd P\notag\\
&=\dfrac{H^2}{P_0^2\cdot\Gamma\left(\frac{H}{P^*}\right)}\int_0^{\infty} e^{-z}~z^{\frac{H}{P^*}-3}~\dd z\notag\\
&=\dfrac{H^2\cdot\Gamma\left(\frac{H}{P^*}-2\right)}{P^2_0\cdot\Gamma\left(\frac{H}{P^*}\right)}=\dfrac{P^{*2}}{P_0^2}\cdot\left(1-\frac{2 P^*}{H}\right)^{-1}\left(1-\frac{P^*}{H}\right)^{-1}.
\end{align}
Hence, we have 
\begin{align}\label{condsqtgteadyvar}
\mathbb{V}ar(R_{\infty}\mid P_0) &=	\E (R_{\infty}\mid P_0)-\E (R_{\infty}\mid P_0)^2\notag\\ &=\dfrac{P^{*2}}{P^2_0}\cdot\left(1-\dfrac{P^*}{H}\right)^{-1}\left[\left(1-\dfrac{2P^*}{H}\right)^{-1}-\left(1-\dfrac{P^*}{H}\right)^{-1}\right]~.	
\end{align}

\subsection{Proof of proposition \ref{prop:super}}\label{A:propsuper}

For the super-exponential growth to persist at least until $t=1$, given that the price started from $P_0$ at $t=0$, 
the following conditions must hold:
\begin{align*}
	\dfrac{\partial}{\partial t}\,&\ln \E(P_{t}\mid P_0)\bigl|_{t=1}>0~,\\
\lim_{t\to 1^-}\dfrac{\partial^2}{\partial t^2}\,&\ln \E(P_{t}\mid P_0)>0~.
\end{align*}
The first order condition (that the expected price grows) is always satisfied. 
We thus focus on the second order condition. Using expression (\ref{E:condret}) for $\E(R_{t}\mid P_0)$, a direct calculation for the 
second-order derivative yields
\begin{align*}
	&\lim_{t\to 1^-}\dfrac{\partial^2}{\partial t^2}\,\ln\E(P_{t}\mid P_0)=\lim_{t\to 1^-}\dfrac{\partial^2}{\partial t^2}\,\ln\E(R_{t}\mid P_0)\\
&=-\biggl\{\,e^{\alpha}\alpha^2\biggl[e^{\alpha+\frac{2H}{(e^{\alpha}-1)P_0}}(e^{\alpha}-1)P_0\left(\dfrac{H}{P_0-e^{\alpha}P_0}\right)^{2q}\\
&\phantom{-\biggl\{\,e^{\alpha}\alpha^2\biggl[}+\left(\Gamma(q)-\Gamma\left(\dfrac{H}{P_0-e^{\alpha}P_0}\right)\right)^2(1+e^{\alpha})H+(e^{\alpha}-1)(q-1)P_0)\\
&\phantom{-\biggl\{\,e^{\alpha}\alpha^2\biggl[}+e^{\frac{H}{(e^{\alpha}-1)P_0}}\left(-\Gamma(q)+\Gamma\left(\dfrac{H}{P_0-e^{\alpha}P_0}\right)\right)\left(\dfrac{H}{P_0-e^{\alpha}P_0}\right)^{q}\\
&\phantom{-\biggl\{\,e^{\alpha}\alpha^2\biggl[++}(e^{\alpha}H+(e^{\alpha}-1)(e^{\alpha}q+1)P_0)\biggr]\biggr\}\biggl/\bigl[\,(e^{\alpha}-1)\left(\Gamma(q)-\Gamma\left(q, -\frac{H}{(e^{\alpha}-1)P_0}\right)\right)^2\,\bigr]~.
\end{align*}
Thus, given that $e^{\alpha}\alpha^2>0$ and $(e^{\alpha}-1)>0$, the second order condition $\lim_{t\to 1^-}\dfrac{\partial^2}{\partial t^2}\,\E(P_{t}\mid P_0)>0$ is equivalent to condition (\ref{supercond}).

The end of the super-exponential growth regime occurs at the time $t_c$ when the log-price exhibits an inflexion point determined by the condition
\begin{align*}
	\dfrac{\partial^2}{\partial t^2}\,&\ln \E(P_{t}\mid P_0)|_{t=t_c}=0~.
\end{align*}
Using expression (\ref{E:condret}) for $\E(R_{t}\mid P_0)$, the above condition leads to Equation (\ref{solvingsupert}). {\hfill $\square$}

. 
\bibliographystyle{elsarticle-harv}
\bibliography{stc-9.bib}

\end{document}